\documentclass[10pt, draftclsnofoot]{IEEEtran}
 \onecolumn
\usepackage{amssymb}
\usepackage{amsmath}
\usepackage[dvips]{graphics}
\usepackage{graphicx}
\usepackage{psfrag}
\usepackage{algorithm}
\usepackage{algpseudocode}
 \allowdisplaybreaks

\usepackage{bbm}
\usepackage{theorem}
\theoremheaderfont{\upshape\bfseries} \theoremstyle{plain}

\theorembodyfont{\rmfamily}

\theoremheaderfont{\itshape}

\newtheorem{definition}{Definition}

\newtheorem{theorem}{Theorem}

\usepackage{lipsum}
\newcommand\blfootnote[1]{%
  \begingroup
  \renewcommand\thefootnote{}\footnote{#1}%
  \addtocounter{footnote}{-1}%
  \endgroup
}

\newcommand{\tab}{\hspace*{.5em}}

\begin{document}

\title{Multiple Access Channels with Combined Cooperation and Partial Cribbing}
\author{Tal Kopetz, Haim Permuter and Shlomo Shamai (Shitz)}
\maketitle
\begin{abstract}
In this paper we study the multiple access channel (MAC) with combined cooperation and partial cribbing and characterize its capacity region. Cooperation means that the two encoders send a message to one another via a rate-limited link prior to transmission, while partial cribbing means that each of the two encoders obtains a deterministic function of the other encoder's output with or without delay. Prior work in this field dealt separately with cooperation and partial cribbing. However, by combining these two methods we can achieve significantly higher rates. Remarkably, the capacity region does not require an additional auxiliary random variable (RV) since the purpose of both cooperation and partial cribbing is to generate a common message between the encoders. In the proof we combine methods of block Markov coding, backward decoding, double rate-splitting, and joint typicality decoding. Furthermore, we present the Gaussian MAC with combined one-sided cooperation and quantized cribbing. For this model, we give an achievability scheme that shows how many cooperation or quantization bits are required in order to achieve a Gaussian MAC with full cooperation/cribbing capacity region. After establishing our main results, we consider two cases where only one auxiliary RV is needed. The first is a rate distortion dual setting for the MAC with a common message, a private message and combined cooperation and cribbing. The second is a state-dependent MAC with cooperation, where the state is known at a partially cribbing encoder and at the decoder. However, there are cases where more than one auxiliary RV is needed, e.g., when the cooperation and cribbing are not used for the same purposes. We present a MAC with an action-dependent state, where the action is based on the cooperation but not on the cribbing. Therefore, in this case more than one auxiliary RV is needed. We deduce a general rule for this result.
\end{abstract}
\section{Introduction}\label{sec:intro}
\blfootnote{This work was supported by the Israel Science Foundation, the ERC starting grant and the European Commission in the framework of the FP7 Network of Excellence in Wireless COMmunications (NEWCOM$\#$).
This paper will be presented in part at the 2014 IEEE International Symposium on Information Theory, Honolulu, HI, USA.
T. Kopetz and H. Permuter are with the department of Electrical and Computer Engineering, Ben-Gurion University of the Negev, Beer-Sheva, Israel (kopetz@post.bgu.ac.il, haimp@bgu.ac.il). 
S. Shamai (Shitz) is with the Department of Electrical Engineering, Technion Institute of Technology, Technion City, Haifa 32000, Israel (sshlomo@ee.technion.ac.il).}
The MAC with cooperating encoders was first studied by Willems \cite{Willems}\nocite{willems1983discrete}-\cite{willems1985discrete}.
Willems introduced two separate approaches to cooperating encoders; in the first, using a rate-limited cooperation link between the two encoders, the two encoders cooperate and share as much of their private messages as possible, while in the second, each encoder "listens" to the other encoder and obtains its output. The second approach was named cribbing. Capacity regions for the two approches, separately, were established by Willems.
Furthermore, the cribbing setting was generalized in \cite{asnani2013multiple} to partial cribbing which means that each of the two encoders obtains a deterministic function of the other encoder’s output. The partial cribbing is especially important in the continuous alphabet, such as the Gaussian MAC, since in a continuous alphabet perfect cribbing means full cooperation between the encoders regardless of the cribbing delay.

In this paper, we combine cooperation and partial cribbing and use them simultaneously, thus obtaining better performance and a larger capacity region.
A MAC with combined cooperation and partial cribbing is depicted in Fig. \ref{fig:main}. Encoder 1 and Encoder 2 obtain messages $M_{21}$ and $M_{12}$ prior to transmission. For the cribbing part, we address two cases. In Case A, the cribbing is done strictly causally by both encoders, i.e., $X_{1,i}$ is a function of $(M_{21},Z_2^{i-1})$ and $X_{2,i}$ is a function of $(M_{12},Z_1^{i-1})$. In Case B, the cribbing is done strictly causally by Encoder 1 and causally by Encoder 2, i.e., $X_{1,i}$ is a function of $(M_{21},Z_2^{i-1})$ and $X_{2,i}$ is a function of $(M_{12},Z_1^{i})$. The idea is that this deterministic function, $Z_1$, is on a sliding scale where one end is $Z_{1,i} = X_{1,i}$ (the actual output) and the other end is when $Z_{1,i}$ is a constant, which does not give any information about $X_{1,i}$. The same applies for $Z_2$. In this research, it was our goal to obtain a generic capacity region for a scheme with both cooperation and partial cribbing.
%

\begin{figure}[h]
\begin{center}
\begin{psfrags}
    \psfragscanon
    \psfrag{A}[][][1]{Encoder 1}
    \psfrag{B}[][][1]{Encoder 2}
    \psfrag{C}[][][1]{$Z_{1,i}=g_1(X_{1,i})$}
    \psfrag{D}[][][1]{$Z_{2,i}=g_2(X_{2,i})$}
    \psfrag{E}[][][1]{$P_{Y|X_1,X_2}$}
    \psfrag{F}[][][1]{Decoder}
    \psfrag{G}[r][][1]{$m_2 \in 2^{nR_2}$}
    \psfrag{H}[r][][1]{$m_1 \in 2^{nR_1}$}
    \psfrag{I}[][][1]{$X_{1,i}(m_1,m_{21},Z_2^{i-1})$}
    \psfrag{J}[][][1]{$X_{2,i}(m_2,m_{12},Z_1^{i-1})$}
    \psfrag{K}[][][1]{$Y_i$}
    \psfrag{L}[][][1]{$\hat{m_1},\hat{m_2}$}
    \psfrag{M}[][][1]{$C_{12}$}
    \psfrag{N}[][][1]{$C_{21}$}
\includegraphics[scale = 0.8]{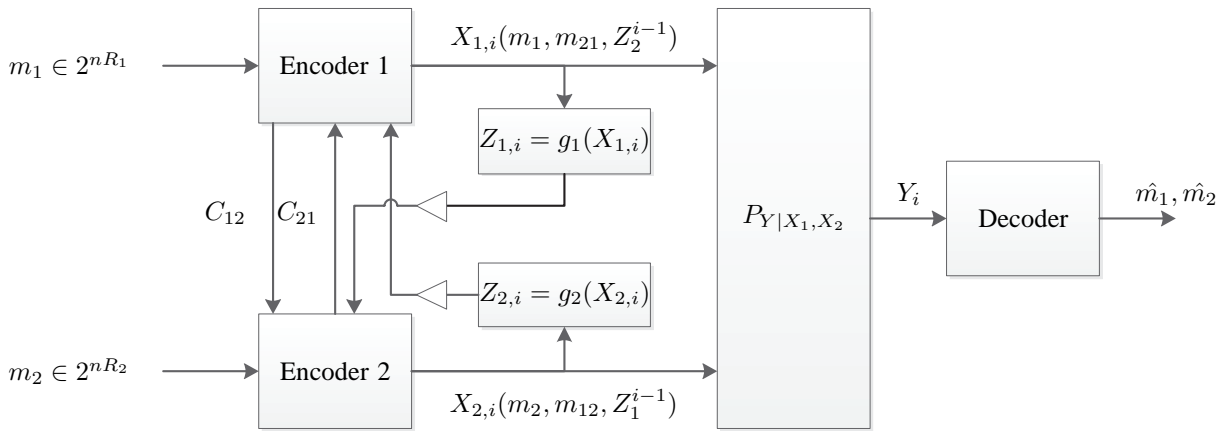}
\caption{MAC with combined cooperation and partial cribbing. Encoder 1 and Encoder 2 obtain messages $M_{21}$ and $M_{12}$ prior to transmission. The cribbing is done strictly causally by both encoders. This setting corresponds to Case A.} \label{fig:main}
\psfragscanoff
\end{psfrags}
\end{center}
\end{figure}
Cooperation and cribbing carry practical implications. In \cite[Chapter 8]{simeone2012cooperative}, Simone et. al. considered cooperative wireless cellular systems and analyzed their performance with separate cooperation and cribbing (referred to as Out-of-Band cooperation and In-Band cooperation, respectively). The results show how cooperation and cribbing separately increase capacity in wireless cellular systems. In the expected 3GPP Release 12, a standard called Proximity Services (ProSE) will be added to the LTE-Advanced "grab bag" of technologies \cite{lin2013overview}. The ProSE protocol will address issues of spectrum utilization, overall throughput, and energy consumption, while enabling new peer to peer and location based applications and services, all of which will be applied using cooperation between "nearby" users in the network. The communication between the users can be attained by using mobile ad hoc networks (Out-of-Band/Cooperation) or by using the same band as the cell sites (In-Band/Cribbing).
Settings of combined cooperation and cribbing considered in this paper give the fundamental limits and insights on how to design optimal coding for communication systems where the users have cognition capabilities and, therefore, "listen" to each other's signals and, in addition, cooperate with each other via dedicated links. We show that combining cribbing and cooperation is straightforward since it does not require any additional auxiliary RV compared with only cribbing or only cooperation.
Therefore, the combination of cooperation and cribbing should be considered in future cooperative wireless cellular systems such as ProSE.\nocite{cao2010cognitive}\nocite{shimonovich2013cognitive}

In this paper, we solve the general model that incorporates both cooperation and partial cribbing.
The capacity regions that were found for cooperation and partial cribbing, separately, in \cite{Willems} and \cite{asnani2013multiple} were constructed using an auxiliary RV, $U$.
That RV signified the information that both encoders share.
In \cite{slepian1973coding}, Slepian and Wolf discovered that the capacity region for the MAC is larger if the encoders share a common message. Therefore, we can refer to the information obtained via cooperation and cribbing as common information shared by both encoders.
One of the results in our work is that the combination of the models does not require an additional auxiliary RV; it is possible to use only one auxiliary RV that represents the common information.
This implies that if for the MAC with partial cribbing we have a "good code", namely, a code that achieves the capacity region, then by performing minor modifications, namely, increasing the common message rate, we can construct a "good code" for the MAC with combined cooperation and partial cribbing.
The coding techniques we use in this paper include block Markov coding (introduced by Willems), joint typicality decoding, backward decoding, and double rate splitting.
Double rate splitting is necessary since we need to split the original message twice; one part will be obtained through the cooperation link and the other part will be obtained using partial cribbing.

Combining cooperation and cribbing was first considered by Bracher and Lapidoth \cite{bracher2012journal} in the context of feedback and state information. However, only strictly-causal perfect cribbing was considered and in our paper we consider partial cribbing both causal and strictly-causal.

After establishing our main results, we present the Gaussian MAC with combined one-sided cooperation and partial cribbing. One can see that an outer bound for the capacity region of this setting is when Encoder 2 knows the message of Encoder 1. Inspired by the work of Asnani et al. \cite{asnani2013multiple} and Bross et al. \cite{bross2008gaussian}, we describe an achievability scheme that coincides with this outer bound in some cases.

Additionaly, we provide a duality between a MAC with a common message, a private message and combined cooperation and cribbing and the rate distortion model known as "Successive Refinement (SR) With Decoder Cooperation" presented in \cite{asnani2013successive}. The decoder cooperation is through a dedicated link and partial cribbing. In this paper we combine both cooperation and partial cribbing in the SR problem and obtain a rate region with only one auxiliary RV.

We go on to study the impact of cooperation and cribbing on state-dependent MACs where the state may provide a refined characterization of the channel, as state-dependent channels are widely studied in the literature.
We address two different state-dependent MACs with cooperation and cribbing (see \cite{bracher2012journal}, \cite{steinberg2012cooperative} for further reading).
The first is a MAC with cooperation and channel state known non-causally at a partially cribbing encoder and at the decoder. In this case we use our results to find a solution with a lone auxiliary RV. Only one auxiliary RV is needed since the purpose of both cooperation and partial cribbing is to generate a common message between the encoders.
The second is a MAC where action-dependent state is known non-causally at a cribbing encoder. Additionally, a one-sided cooperation link is attained at the cribbing encoder. Action-dependent states were introduced by Weissman in \cite{weissman2010capacity}. The action is based on the private message of the cribbing encoder and the message from the cooperation link. In this case, a lone auxiliary RV will not suffice since the purpose of the cooperation is not only to generate a common message but also to contribute to the action and affect the channel state.

The remainder of the paper is organized as follows:
In Section \ref{sec:def+results}, we define the MAC with combined cooperation and partial cribbing and provide its capacity region for two cases. The first is for strictly causal partial cribbing (Case A) and the second is for mixed causal and strictly causal partial cribbing (Case B). Thereafter, the proof for both cases is provided.
In Section \ref{sec:gaussian}, we give an achievability scheme for the Gaussian MAC with combined one-sided cooperation and partial cribbing.
In Section \ref{sec:duality}, we establish the duality between the MAC with combined cooperation and partial cribbing at the encoders and the SR problem with combined cooperation and partial cribbing at the decoders. We show that a lone RV is needed to characterize the rate region of the SR problem.
In Section \ref{sec:state}, we give an example of a state-dependent MAC with combined cooperation and partial cribbing where only one auxiliary RV is needed.
In Section \ref{sec:action}, we study the case of the MAC with an action-dependent state where more than one auxiliary RV is needed and consider its implications.
In Section \ref{sec:conc} we conclude the paper and suggest some research directions that have not yet been solved such as noncausal partial cribbing and combined cooperation and cribbing in the interference channel.

\section{The MAC with Combined Cooperation and Partial Cribbing}\label{sec:def+results}
\subsection{Definitions and Main Results}
Let us consider the MAC with combined cooperation and partial cribbing depicted in Fig. \ref{fig:main}.
The MAC setting consists of two transmitters (encoders) and one receiver (decoder).
Each transmitter $l \in \{1,2\}$ chooses an index $m_l$ uniformly from the set $\{1,\dots,2^{nR_l}\}$ and independently of the other transmitter.
The input to the channel from Encoder $l \in \{1,2\}$ is denoted by $\{X_{l,1},X_{l,2},X_{l,3},\dots\}$.
Encoder 1 and Encoder 2 obtain deterministic functions of the form $Z_{2,i} = g_2(X_{2,i})$ and $Z_{1,i} = g_1(X_{1,i})$, respectively. We address two cases in this setting:
\begin{itemize}
  \item Case A : Both Encoder 1 and Encoder 2 obtain $Z_{2,i}$ and $Z_{1,i}$, respectively, with unit delay.
  \item Case B : Encoder 1 obtains $Z_{2,i}$ with unit delay and Encoder 2 obtains $Z_{1,i}$ without delay.
\end{itemize}
Additionally, Encoder 1 obtains a message $m_{21} \in \{1,\dots,2^{nC_{21}}\}$ from Encoder 2 and Encoder 2 obtains a message $m_{12} \in \{1,\dots,2^{nC_{12}}\}$ from Encoder 1. Both messages are obtained prior to the transmission of $(X_1^n,X_2^n)$ through the channel.
The output of the channel is denoted by $\{Y_1,Y_2,Y_3,\dots\}$.
The channel is characterized by a conditional probability $P(y_i|x_{1,i},x_{2,i})$.
The channel probability does not depend on the time index $i$ and is memoryless, i.e.,
\begin{eqnarray}
P(y_i|x_1^i,x^i_2,y^{i-1}) = P(y_i|x_{1,i},x_{2,i}),\label{eq:DMC}
\end{eqnarray}
where the superscripts denote sequences in the following way: $x^i_l=(x_{l,1},x_{l,2},\dots,x_{l,i}),l\in\{1, 2\}$. Since the settings in this paper do not include feedback from the receiver to the transmitters, i.e., $P(x_{1,i},x_{2,i}|x^{i-1}_1,x^{i-1}_2,y^{i-1})=P(x_{1,i},x_{2,i}|x^{i-1}_1,x^{i-1}_2)$, equation (\ref{eq:DMC}) implies that
\begin{eqnarray}
P(y_i|x_1^n,x^n_2,y^{i-1}) = P(y_i|x_{1,i},x_{2,i}).
\end{eqnarray}
\begin{definition}\label{def:main}
A $(2^{nR_1},2^{nR_2},2^{nC_{12}},2^{nC_{21}},n)$ \textit{code} for the MAC with combined cooperation and partial cribbing, as shown in Fig. \ref{fig:main}, consists at time $i$ of encoding functions at Encoder 1 and Encoder 2
\begin{eqnarray}
f_{12}&:&\{1,\dots,2^{nR_1}\}\mapsto\{1,\dots,2^{nC_{12}}\},\\
f_{21}&:&\{1,\dots,2^{nR_2}\}\mapsto\{1,\dots,2^{nC_{21}}\},\\
f_{1,i}&:&\{1,\dots,2^{nR_1}\}\times\{1,\dots,2^{nC_{21}}\}\times\mathcal{Z}^{i-1}_2 \mapsto  \mathcal{X}_{1,i},\\
f^A_{2,i}&:&\{1,\dots,2^{nR_2}\}\times\{1,\dots,2^{nC_{12}}\}\times\mathcal{Z}^{i-1}_1 \mapsto  \mathcal{X}_{2,i},\\
f^B_{2,i}&:&\{1,\dots,2^{nR_2}\}\times\{1,\dots,2^{nC_{12}}\}\times\mathcal{Z}^{i}_1 \mapsto  \mathcal{X}_{2,i},
\end{eqnarray}
and a decoding function
\begin{eqnarray}
g:\mathcal{Y}^n \mapsto \{1,\dots,2^{nR_1}\}\times\{1,\dots,2^{nR_2}\}.
\end{eqnarray}
The average probability of error for a $(2^{nR_1},2^{nR_2},2^{nC_{12}},2^{nC_{21}},n)$ code is defined as
\begin{eqnarray}
P^{(n)}_e = \frac{1}{2^{n(R_1+R_2)}} \sum_{m_1,m_2} \Pr\{g(Y^n)\ne(m_1,m_2)|(m_1,m_2)\ \text{sent}\}.
\end{eqnarray}
\end{definition}
A rate $(R_1,R_2)$ is said to be \textit{achievable} for the MAC with combined cooperation and partial cribbing if there exists a sequence of $(2^{nR_1},2^{nR_2},2^{nC_{12}},2^{nC_{21}},n)$ codes s.t. $P^{(n)}_e \rightarrow 0$.
The \textit{capacity region} of the MAC is the closure of all achievable rates. The following theorem describes the capacity region of a MAC with combined cooperation and partial cribbing.

Let us define the following regions, $\mathcal{R}^A$ and $\mathcal{R}^B$, that are contained in $ \mathbb{R}_+^2$, namely, contained in the set of nonnegative two-dimensional real numbers.
\begin{eqnarray}
\mathcal{R}^A =\left\{
                 \begin{array}{c}
                   R_1 \leq I(X_1; Y |X_2, Z_1, U)+H(Z_1|U) + C_{12},\\
                   R_2 \leq I(X_2; Y |X_1, Z_2, U)+H(Z_2|U) + C_{21},\\
                   R_1 + R_2 \leq I(X_1,X_2; Y |U,Z_1, Z_2) + H(Z_1, Z_2|U)+C_{12}+C_{21},\\
                   R_1 + R_2 \leq I(X_1,X_2; Y ) $, for$\\
                   P(u)P(x_1|u)\mathbbm{1}_{z_1=f(x_1)}P(x_2|u)\mathbbm{1}_{z_2=f(x_2)}P(y|x_1, x_2).\\
                 \end{array}
               \right\}\label{eq:capacity}.
\end{eqnarray}
The region $\mathcal{R}^B$ is defined with the same set of inequalities as in (\ref{eq:capacity}), but the joint distribution is of the form
\begin{eqnarray}
                   P(u)P(x_1|u)\mathbbm{1}_{z_1=f(x_1)}P(x_2|u, z_1)\mathbbm{1}_{z_2=f(x_2)}P(y|x_1, x_2).
\end{eqnarray}
\begin{theorem}\label{theorem:MAC+crib+coop}{\emph{(Capacity Region of the MAC with Combined Cooperation and Partial Cribbing)}}
The capacity regions of the MAC with combined cooperation and strictly causal (Case A) and mixed strictly causal and causal (Case B) partial cribbing, as described in Def. \ref{def:main}, are $\mathcal{R}^A$ and $\mathcal{R}^B$, respectively.
\end{theorem}
We note that $H(Z_1|U) = I(Z_1;X_1|U)$; thus the cribbing, $I(Z_1;X_1|U)$, plays the same role (in a quantitative sense) to the cooperation link, $C_{12}$.
Similarly, the role of $I(Z_2;X_2|U)$ to $C_{21}$ and of $I(Z_1,Z_2;X_1,X_2|U)$ to $C_{12}+C_{21}$.
Hence, the important feature is the mutual information of the cooperation, whether the cooperation is done by cribbing or by dedicated links, and they both act in a similar way.

A straightforward result from Theorem \ref{theorem:MAC+crib+coop} is the capacity region for the compound MAC \cite{ahlswede1974capacity} with combined cooperation and partial cribbing. The region and proof for the compound MAC are omitted for brevity.

\subsection{Proof of Theorem \ref{theorem:MAC+crib+coop}}\label{sec:proof}
\subsubsection{Converse}
We will start with the converse of Case A.

\textit{Converse for Case A:}
Given an achievable rate $(R_1,R_2)$ we need to show that there exists a joint distribution of the form $P(u)P(x_1|u)\mathbbm{1}_{z_1=f(x_1)}P(x_2|u)\mathbbm{1}_{z_2=f(x_2)}P(y|x_1, x_2)$ such that the inequalities (\ref{eq:capacity}) are satisfied. Since $(R_1,R_2)$ is an achievable rate-pair, there exists a $(2^{nR_1},2^{nR_2},2^{nC_{12}},2^{nC_{21}},n)$ code with an arbitrarily small error probability $P^{(n)}_e$. By Fano's inequality,
\begin{eqnarray}
H(M_1,M_2|Y^n)\leq n(R_1 + R_2)P^{(n)}_e + H(P^{(n)}_e).
\end{eqnarray}
We set
\begin{equation}
(R_1 + R_2)P^{(n)}_e + \frac{1}{n}H(P^{(n)}_e)\triangleq \epsilon_n,
\end{equation}
where $\epsilon_n \rightarrow 0$ as $P^{(n)}_e\rightarrow 0$. Hence,
\begin{eqnarray}
H(M_1|Y^n,M_2)\leq H(M_1,M_2|Y^n)\leq n\epsilon_n,\\
H(M_2|Y^n,M_1)\leq H(M_1,M_2|Y^n)\leq n\epsilon_n.
\end{eqnarray}
For $R_1$ we have the following:
\begin{eqnarray}
    nR_1 &=& H(M_1)\\
        &\stackrel{(a)}=& H(M_1,M_{12},Z_1^n|M_2)\\
        &\stackrel{(b)}=& H(M_{12}|M_2) + H(Z_1^n|M_{12},M_2) + H(M_1|Z_1^n,M_{12},M_2)\\
        &=& H(M_{12}) + H(Z_1^n|M_{12},M_{21},M_2) + H(M_1|Z_1^n,M_{12},M_2) \notag\\ &&\tab + H(M_1|Y^n,Z_1^n,M_{12},M_2) - H(M_1|Y^n,Z_1^n,M_{12},M_2)\\
        &\stackrel{(c)}\leq& H(M_{12}) + H(Z_1^n|M_{12},M_{21},M_2) + I(M_1;Y^n|Z_1^n,M_{12},M_2,M_{21}) + n\epsilon_n\\
        &\stackrel{(d)}=& H(M_{12}) + \sum_{i=1}^n [H(Z_{1,i}|Z_1^{i-1},M_{12},M_{21},M_2) \notag\\ &&\tab + I(M_1;Y_i|Y^{i-1},Z_1^n,M_{12},M_2,M_{21})] + n\epsilon_n\\
        &\stackrel{(e)}=& H(M_{12}) + \sum_{i=1}^n [H(Z_{1,i}|Z_1^{i-1},Z_2^{i-1},M_{12},M_{21},M_2) \notag\\ &&\tab + I(M_1,X_{1,i};Y_i|Y^{i-1},Z_1^n,Z_2^{i-1},M_{12},M_2,M_{21})] + n\epsilon_n\\
        &\stackrel{(f)}\leq& H(M_{12}) + \sum_{i=1}^n [H(Z_{1,i}|Z_1^{i-1},Z_2^{i-1},M_{12},M_{21}) \notag\\ &&\tab + I(X_{1,i};Y_i|X_{2,i},Z_1^i,Z_2^{i-1},M_{12},M_{21})] + n\epsilon_n,
\end{eqnarray}
where (a) follows since messages $M_1$ and $M_2$ are independent and since $(M_{12},Z_1^n)=f(M_1,M_2)$, (b) and (d) follow from the chain rule, (c) follows from Fano's inequality and because $M_{21}$ is a function of $M_2$, (e) follows since $Z_2^{i-1}$ is a function of $(M_{12},M_2)$ and $X_{1,i}$ is a function of $(M_1,M_{21})$, and step (f) follows since conditioning reduces entropy and from the Markov chain $Y_i-(X_{1,i},X_{2,i},M_{12},M_{21},Z_1^i,Z_2^{i-1})-(M_1,M_2,Y^{i-1})$. From the definition of a RV
\begin{eqnarray}
    U_i\triangleq(Z_1^{i-1},Z_2^{i-1},M_{12},M_{21}),
\end{eqnarray}
we obtain
\begin{eqnarray}
    R_1&\leq& C_{12} + \frac{1}{n}\sum_{i=1}^n [H(Z_{1,i}|U_i) + I(X_{1,i};Y_i|X_{2,i},Z_{1,i},U_i)] + \epsilon_n. \label{R_1}
\end{eqnarray}
Similarly to (\ref{R_1}), we obtain
\begin{eqnarray}
    R_2&\leq& C_{21} + \frac{1}{n}\sum_{i=1}^n [H(Z_{2,i}|U_i) + I(X_{2,i};Y_i|X_{1,i},Z_{2,i},U_i)] + \epsilon_n.
\end{eqnarray}
Now, consider
\begin{eqnarray}
    n(R_1+R_2) &=& H(M_1,M_2)\\
        &\stackrel{(a)}=& H(M_1,M_2,Z_1^n,Z_2^n,M_{12},M_{21})\\
        &\stackrel{(b)}=& H(M_{12}) + H(M_{21}|M_{12}) + H(Z_1^n,Z_2^n|M_{12},M_{21}) \notag\\ &&\tab + H(M_1,M_2|Z_1^n,Z_2^n,M_{12},M_{21})\\
        &\stackrel{(c)}\leq& H(M_{12}) + H(M_{21}) + H(Z_1^n,Z_2^n|M_{12},M_{21}) \notag\\ &&\tab + I(M_1,M_2;Y^n|Z_1^n,Z_2^n,M_{12},M_{21}) + n\epsilon_n\\
        &\stackrel{(d)}\leq& nC_{12} + nC_{21} +  \sum_{i=1}^n [H(Z_{1,i},Z_{2,i}|Z_1^{i-1},Z_2^{i-1},M_{12},M_{21}) \notag\\ &&\tab + I(M_1,M_2;Y_i|Y^{i-1},Z_1^n,Z_2^n,M_{12},M_{21})] + n\epsilon_n\\
        &\stackrel{(e)}=& nC_{12} + nC_{21} + \sum_{i=1}^n [H(Z_{1,i},Z_{2,i}|Z_1^{i-1},Z_2^{i-1},M_{12},M_{21}) \notag\\ &&\tab + I(M_1,X_{1,i},M_2,X_{2,i};Y_i|Y^{i-1},Z_1^n,Z_2^n,M_{12},M_{21})] + n\epsilon_n\\
        &\stackrel{(f)}=& nC_{12} + nC_{21} + \sum_{i=1}^n [H(Z_{1,i},Z_{2,i}|Z_1^{i-1},Z_2^{i-1},M_{12},M_{21}) \notag\\ &&\tab + I(X_{1,i},X_{2,i};Y_i|Z_1^i,Z_2^i,M_{12},M_{21})] + n\epsilon_n,
\end{eqnarray}
where (a) follows from the fact that $(M_{12},M_{21},Z_1^n,Z_2^n)=f(M_1,M_2)$, (b) and (d) follow from the chain rule, (c) follows from Fano's inequality and because $M_{21}$ is independent of $M_{12}$, (e) follows from the fact that $(X_{1,i},X_{2,i})=f(M_1,M_2)$, and step (f) follows from the Markov chain $Y_i-(X_{1,i},X_{2,i},Z_1^i,Z_2^i,M_{12},M_{21})-(M_1,M_2,Y^{i-1})$. From the definition of the RV $U$, we obtain
\begin{eqnarray}
    R_1+R_2&\leq& C_{12} + C_{21} + \frac{1}{n}\sum_{i=1}^n [H(Z_{1,i},Z_{2,i}|U_i) + I(X_{1,i},X_{2,i};Y_i|Z_{1,i},Z_{2,i},U_i)] + \epsilon_n.
\end{eqnarray}
Furthermore, consider
\begin{eqnarray}
    n(R_1+R_2) &=& H(M_1,M_2)\\
        &=& H(M_1,M_2) + H(M_1,M_2|Y^n) - H(M_1,M_2|Y^n)\\
        &\stackrel{(a)}\leq& I(M_1,M_2;Y^n) + n\epsilon_n\\
        &\stackrel{(b)}=& I(X_1^n,X_2^n;Y^n) + n\epsilon_n\\
        &\stackrel{(c)}=& \sum_{i=1}^n I(X_1^n,X_2^n;Y_i|Y^{i-1}) + n\epsilon_n\\
        &\stackrel{(d)}=& \sum_{i=1}^n I(X_{1,i},X_{2,i};Y_i) + n\epsilon_n,
\end{eqnarray}
where (a) follows from Fano's inequality, (b) follows from the fact that $(X_1^n,X_2^n)$ is a deterministic function of $(M_1,M_2)$ and from the Markov chain $Y^n-(X_1^n,X_2^n)-(M_1,M_2)$, (c) follows from the chain rule, and step (d) follows from the memoryless property of the channel. Thus we obtain
\begin{eqnarray}
        R_1+R_2 \leq \frac{1}{n}\sum_{i=1}^n I(X_{1,i},X_{2,i};Y_i) + \epsilon_n.
\end{eqnarray}

Finally, we will prove the following Markov chains:
\begin{itemize}
\item $Z_{2,i} - U_i - Z_{1,i}$ - We will prove this graphically as in \cite[Section II]{permuter2010two}. Using the undirected graph in Fig. \ref{fig:markov-A}, we can see that the Markov Chain $Z_{2,i} - (M_{12},M_{21},Z_1^{i-1},Z_2^{i-1}) - Z_{1,i}$ holds since we cannot get from node $Z_{2,i}$ to node $Z_{1,i}$ without going through nodes $(M_{12},M_{21},Z_1^{i-1},Z_2^{i-1})$.
\begin{figure}[t]
    \begin{center}
        \begin{psfrags}
            \psfragscanon
            \psfrag{A}[][][1]{$M_1$}
            \psfrag{B}[][][1]{$M_{21}$}
            \psfrag{C}[][][1]{$X_{1,i}$}
            \psfrag{D}[][][1]{$Z_1^{i-1}$}
            \psfrag{E}[][][1]{$M_{12}$}
            \psfrag{F}[][][1]{$Z_2^{i-1}$}
            \psfrag{G}[][][1]{$M_2$}
            \psfrag{H}[][][1]{$X_{2,i}$}
            \psfrag{I}[][][1]{$Z_{1,i}$}
            \psfrag{J}[][][1]{$Z_{2,i}$}
            \includegraphics[scale = 0.8]{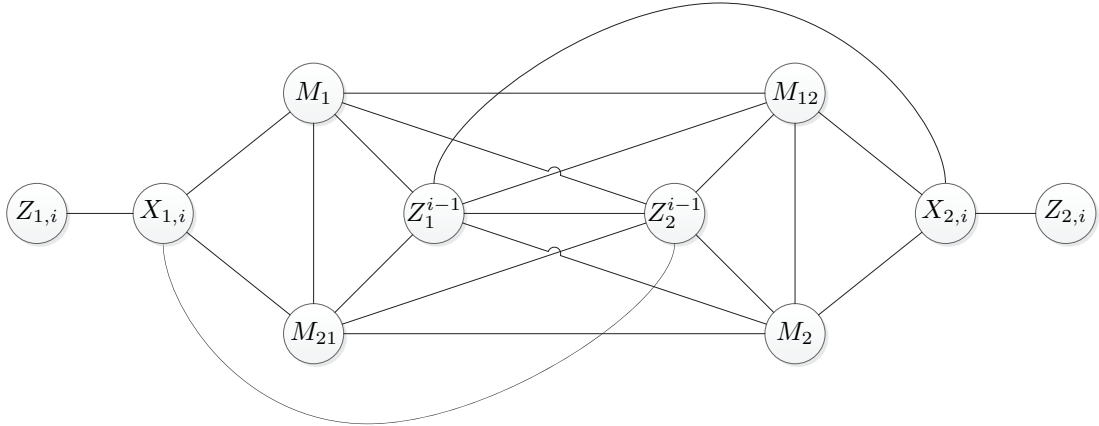}
            \caption{Proof of the Markov Chain $X_{2,i} - (M_{12},M_{21},Z_1^{i-1},Z_2^{i-1}) - X_{1,i}$ using the undirected graphical technique \cite[Section II]{permuter2010two}. This graph corresponds to the joint distribution $P(m_1)P(m_2)P(m_{12}|m_1)P(m_{21}|m_2)\prod_{k=1}^{i-1}P(z_{1,k}|m_1,m_{21},z_2^{k-1})P(z_{2,k}|m_2,m_{12},z_1^{k-1})P(x_{1,i}|m_{21},m_1,z_2^{i-1})$ $P(x_{2,i}|m_{12},m_2,z_1^{i-1})P(z_{1,i}|x_{1,i})P(z_{2,i}|x_{2,i})$.} \label{fig:markov-A}
            \psfragscanoff
        \end{psfrags}
    \end{center}
    \vspace{3mm}
\end{figure}
\item $X_{1,i} - (U_i,Z_{1,i}) - Z_{2,i}$ - Using the undirected graph in Fig. \ref{fig:markov-A}, we can see that the Markov Chain $X_{1,i} - (M_{12},M_{21},Z_1^{i},Z_2^{i-1}) - Z_{2,i}$ holds since we cannot get from node $X_{1,i}$ to node $Z_{2,i}$ without going through nodes $(M_{12},M_{21},Z_1^{i},Z_2^{i-1})$.
\item $X_{2,i} - (U_i,Z_{2,i}) - X_{1,i}$ - Using the undirected graph in Fig. \ref{fig:markov-A}, we can see that the Markov Chain $X_{2,i} - (M_{12},M_{21},Z_1^{i-1},Z_2^{i}) - X_{1,i}$ holds since we cannot get from node $X_{2,i}$ to node $X_{1,i}$ without going through nodes $(M_{12},M_{21},Z_1^{i-1},Z_2^{i})$.
\item $Y_i - (X_{1,i},X_{2,i}) - (Z_{1,i},Z_{2,i},U_i)$ - Follows since the channel output at time $i$ depends on the history $(X_1^i,X_2^i)$ only through $(X_{1,i},X_{2,i})$.
\end{itemize}
Finally, let $Q$ be an RV independent of $(X_1^n,X_2^n,Y^n)$ and uniformly distributed over the set $\{1,2,3,\dots,n\}$. We define the RVs $U\triangleq(Q,U_Q), X_1\triangleq X_{1,Q}, X_2\triangleq X_{2,Q}$, and $Y\triangleq Y_Q$ to obtain the region given in (\ref{eq:capacity}). This completes the converse for Case A.\hfill $\blacksquare$

\textit{Converse for Case B:} We repeat the same approach as for Case A, except that in the final step we need to show the Markov chain $X_{2,i}-(U_i, Z_{1,i}, Z_{2,i})-X_{1,i}$ rather than $X_{2,i}-(U_i, Z_{2,i})-X_{1,i}$ as in Case A. Since for Case A $X_{2,i} - (M_{12},M_{21},Z_1^{i-1},Z_2^{i}) - X_{1,i}$ holds, then $X_{2,i} - (M_{12},M_{21},Z_1^{i},Z_2^{i}) - X_{1,i}$ also holds.\hfill $\blacksquare$

\subsubsection{Achievability}

\textit{Achievability for Case A:} To prove the achievability of the capacity region, we need to show that for a fixed distribution of the form $P(u)P(x_1|u)\mathbbm{1}_{z_1=f(x_1)}P(x_2|u)\mathbbm{1}_{z_2=f(x_2)}P(y|x_1, x_2)$ and for $(R_1,R_2)$ that satisfy the inequalities in (\ref{eq:capacity}), there exists a sequence of $(2^{nR_1},2^{nR_2},2^{nC_{12}},2^{nC_{21}},n)$ codes for which $P^{(n)}_e\rightarrow0$ as $n\rightarrow\infty$.

The idea behind this proof is to convert the cooperation problem into a setting that corresponds to the MAC with a common message and partially cribbing encoders considered in \cite{asnani2013multiple} and rely on its capacity region to show that the cooperation capacity region is indeed achievable.
This is done by sharing as much as possible of the original private messages, $(m_1,m_2)$, through the communication links in order to create a common message;
the unshared parts of the original messages serve as the new private messages.
By doing so, the coding scheme of the setting with a common message can be employed.
The capacity region found in \cite{asnani2013multiple} for the MAC with a common message and partially cribbing encoders is
\begin{eqnarray}
    \tilde{R_1} &\leq& H(Z_1|U) + I(X_1; Y |X_2, Z_1, U), \notag\\
    \tilde{R_2} &\leq& H(Z_2|U) + I(X_2; Y |X_1, Z_2, U), \notag\\
    \tilde{R_1} + \tilde{R_2} &\leq& I(X_1,X_2; Y |U,Z_1, Z_2) + H(Z_1, Z_2|U), \notag\\
    \tilde{R_0} + \tilde{R_1} + \tilde{R_2} &\leq& I(X_1,X_2; Y).\label{eq:capacity-mac-ce}
\end{eqnarray}
The achievability proof for the MAC with a common message and partially cribbing encoders is available in Appendix \ref{appendix:achieve}.
Let us define the following rates
\begin{eqnarray}
    \tilde{R_0} = C_{12} + C_{21},\\
    \tilde{R_1} = R_1 - C_{12},\\
    \tilde{R_2} = R_2 - C_{21},
\end{eqnarray}
i.e., we defined the common message as the messages that are transmitted through the cooperation links.
With respect to these definitions, the inequalities in (\ref{eq:capacity-mac-ce}) can be rewritten as
\begin{eqnarray}
    R_1 - C_{12} &\leq& H(Z_1|U) + I(X_1; Y |X_2, Z_1, U), \notag\\
    R_2 - C_{21} &\leq& H(Z_2|U) + I(X_2; Y |X_1, Z_2, U), \notag\\
    (R_1 - C_{12}) + (R_2 - C_{21}) &\leq& I(X_1,X_2; Y |U,Z_1, Z_2) + H(Z_1, Z_2|U), \notag\\
    (C_{12} + C_{21}) + (R_1 - C_{21}) + (R_2 - C_{21}) &\leq& I(X_1,X_2; Y ), \label{eq:capacity_old}
\end{eqnarray}
which is equivalent to the region in (\ref{eq:capacity}).\hfill $\blacksquare$

\textit{Achievability for Case B:} The achievability of case B is very similar to that of case A, only the codewords of $X_2$ need to be generated according to Shannon's strategy (or a code-tree) rather than codewords. This is due to the fact that $Z_{1,i}$ is known causally and $X_2$ is generated according to a distribution $P(x_2|u, z_1, z_2)$.
\hfill $\blacksquare$

\section{Gaussian MAC with Combined Cooperation and Quantized Cribbing}\label{sec:gaussian}
We now consider a Gaussian MAC, i.e., $Y = X_1+X_2+W$ where $W\thicksim N(0,N)$, depicted in Fig. \ref{fig:gaussian}.
\begin{figure}[h]
    \begin{center}
        \begin{psfrags}
            \psfragscanon
            \psfrag{A}[][][1]{Encoder 1}
            \psfrag{B}[][][1]{Encoder 2}
            \psfrag{C}[][][1]{Quantizer}
            \psfrag{D}[][][1]{$Z_i$}
            \psfrag{F}[][][1]{Decoder}
            \psfrag{G}[r][][1]{$m_2 \in 2^{nR_2}$}
            \psfrag{H}[r][][1]{$m_1 \in 2^{nR_1}$}
            \psfrag{I}[][][1]{$X_{1,i}(m_1)$}
            \psfrag{J}[][][1]{$X_{2,i}(m_2,m_{12},Z^{i})$}
            \psfrag{K}[][][1]{$Y_i$}
            \psfrag{L}[][][1]{$\hat{m_1},\hat{m_2}$}
            \psfrag{M}[][][1]{$C_{12}$}
            \psfrag{N}[][][1]{$W\thicksim N(0,N)$}
            \includegraphics[scale = 0.8]{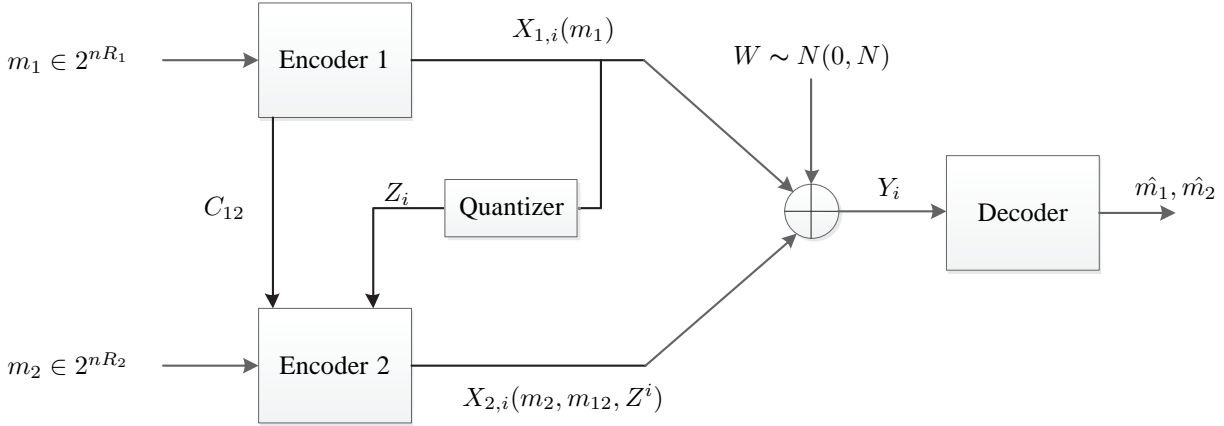}
            \caption{Gaussian MAC with one-sided combined cooperation and quantized cribbing. Message $M_{12}$ is sent prior to transmission and $Z_i$ is known causally at Encoder 2.} \label{fig:gaussian}
            \psfragscanoff
        \end{psfrags}
    \end{center}
\end{figure}

We assume that the power constraints over the outputs of Encoder 1 and Encoder 2 are $P_1$ and $P_2$, respectively. Prior to transmission, Encoder 1 sends a message $M_{12}$ to Encoder 2. In addition, Encoder 2 cribs causally from Encoder 1 and obtains $Z_i$, which is a scalar quantization of the signal $X_{1,i}$.
First, we look at an inner bound to the capacity region, which is the Gaussian MAC without cooperation and cribbing. The capacity region in this case is
\begin{eqnarray}
  R_1 &\leq& \frac{1}{2}\log(1+\frac{P_1}{N}),\notag\\
  R_2 &\leq& \frac{1}{2}\log(1+\frac{P_2}{N}),\notag\\
  R_1 + R_2 &\leq& \frac{1}{2}\log(1+\frac{P_1+P_2}{N}).
\end{eqnarray}
On the other hand, an outer bound is obtained when there is full cooperation or perfect cribbing, i.e., Encoder 2 obtains the message $m_1$ before sending $X_2$. The capacity region in this case is
\begin{eqnarray}
  R_2 &\leq& \frac{1}{2}\log(1+\frac{P_2}{N}(1-\rho^2)),\notag\\
  R_1 + R_2 &\leq& \frac{1}{2}\log(1+\frac{P_1+2\rho\sqrt{P_1P_2}+P_2}{N}).
\end{eqnarray}
We now present an achievability scheme inspired by the work of Asnani et al. \cite{asnani2013multiple} and Bross et al. \cite{bross2008gaussian}. In \cite{asnani2013multiple}, an achievable region for the Gaussian MAC with quantized cribbing has been described, whereas in \cite{bross2008gaussian}, an achievable region for the Gaussian MAC with a common message was provided. In our work, we combine the two achievability schemes. We set the following distributions:
\begin{eqnarray}
  X_1 &=& \lambda U + X^\prime_1,\\
  X_2 &=& \bar{\lambda} U + X^\prime_2,
\end{eqnarray}
where
\begin{eqnarray}
  U \thicksim N(0,P_0) &,& P_0 = \left(\sqrt{\bar{\beta_1}P_1} + \sqrt{\bar{\beta_2}P_2}\right)^2, \notag\\
  P_{X_2^{\prime}|Z,U}(x^{\prime}_2|z,u) &=& \bar{\rho}P_{X^{\prime\prime}_2}(x^{\prime}_2) + \rho P_{X_1^{\prime}|Z,U}(x^{\prime}_2|z,u),\notag\\
  X^\prime_1 &\thicksim& N(0,\beta_1P_1),\notag\\
  X^{\prime\prime}_2 &\thicksim& N(0,\beta_2P_2),\notag\\
  \lambda = \sqrt{\frac{\bar{\beta_1}P_1}{P_0}}&,& \bar{\lambda} = 1 - \lambda,\notag\\
  \beta_1,\beta_2,\rho &\in& [0,1].
\end{eqnarray}
The intuition behind the choice of these distributions is as follows. The common message, signified as $U$, is obtained via the rate-limited link and the two encoders cooperate to send that common message. Since the cooperation and cribbing are one-sided, only Encoder 2 can help Encoder 1 send his private message. The idea behind the choice of $P_{X_2^{\prime}|Z,U}(x^{\prime}_2|z,u)$ is that Encoder 2 will send $\bar{\rho}$ of the time his private message and $\rho$ of the time the estimation of Encoder 1's private message, $X_1^\prime$, conditioned on the cribbing $Z$ and the cooperation $U$.\begin{figure}[!t]
    \begin{center}
        \begin{psfrags}
            \psfragscanon
            \psfrag{A}[][][0.7]{1-bit quantization, $C_{12}=0.4$}
            \psfrag{B}[][][0.7]{$R_1$}
            \psfrag{C}[][][0.7]{$R_2$}
            \psfrag{H}[][][0.7]{2-bit quantization, $C_{12}=0.4$}
            \psfrag{I}[][][0.7]{$R_1$}
            \psfrag{J}[][][0.7]{$R_2$}
            \includegraphics[scale = 0.6]{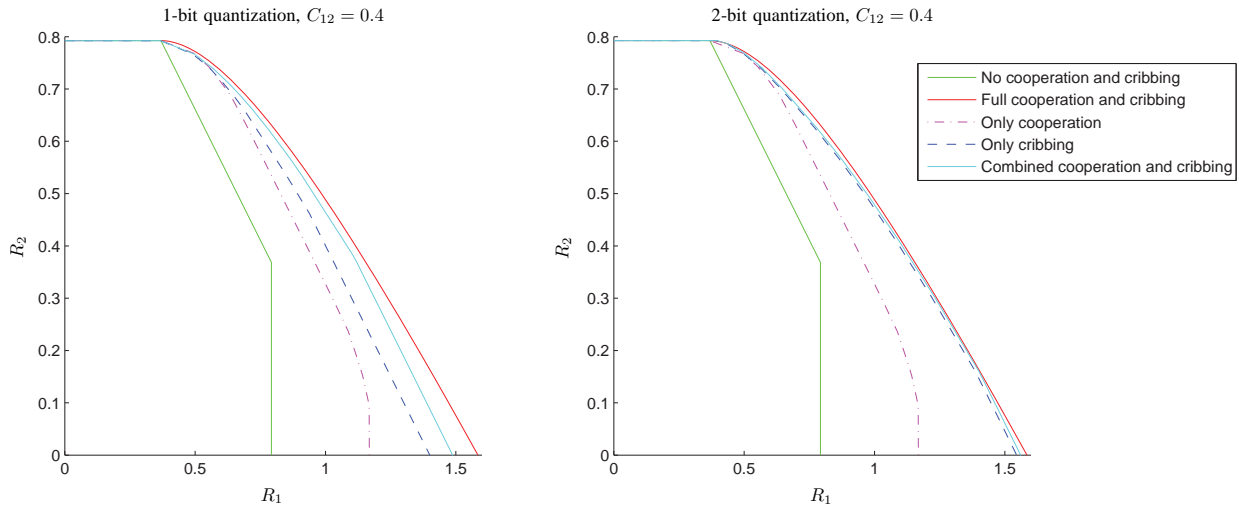}
            \caption{Achievable regions for the Gaussian MAC with combined cooperation and quantized cribbing.} \label{fig:quantized-graph}
            \psfragscanoff
        \end{psfrags}
    \end{center}
    \vspace{-8mm}
\end{figure}
Notice that under these definitions, by setting the power constraints as $P_1 = P_2 = 1$, the power constraints on both encoders hold. Evaluation of region $\mathcal{R}_B$ with $Z_2$ constant and $N = \frac{1}{2}$ is depicted in Fig. \ref{fig:quantized-graph}; achievable regions for 1-bit and 2-bit quantizations are illustrated where $C_{12}=0.4$. When only one bit of quantization is available (LHS of Fig. \ref{fig:quantized-graph}), the region of combined cooperation and cribbing encloses special cases of cribbing \cite{asnani2013multiple} and cooperation \cite{bross2008gaussian}. However, when two bits of quantization are available (RHS of Fig. \ref{fig:quantized-graph}), combining cooperation and cribbing does not significantly increase the region. This is because the difference between the achievable region with a 2-bit quantizer ($C_{12}=0$) and full cooperation is negligible.

\section{Dual Rate Distortion Setting}\label{sec:duality}
The information-theoretic duality between rate distortion and channel coding was first introduced by Shannon in \cite{shannon1959coding}. An important duality between the Wyner-Ziv rate distortion problem \cite{wyner1976rate} and the Gelfand-Pinsker channel coding problem \cite{gel1980coding} was pointed out by Cover and Chiang in \cite{cover2002duality} (see \cite{pradhan2003duality} and \cite{gupta2011operational} for further reading). In some cases, the corner points of a rate distortion region and its dual channel coding capacity region are the same. This property can help one find a region based on its dual region.
In general, there is no solution for the dual setting of the MAC. However, the rate distortion dual of the MAC with a common message has been solved. In \cite{asnani2013successive}, Asnani et. al. considered the SR problem with decoder cooperation and its channel coding duals.
In this section we show how our methods of combined cooperation and cribbing can be implemented in the rate distortion dual. We establish the duality between the MAC with a common message, a private message, and combined cooperation and partial cribbing and the SR problem with combined cooperation and partial cribbing at the decoder. As expected, the rate region for the rate distortion dual consists of a single RV. \begin{table}[!h]
\vspace{-5mm}
\begin{center}
  \caption{Principles of duality between channel coding and source coding}\label{table:principles}
  \begin{tabular}{|| c | c ||}
    \hline \hline
    \large{Channel coding} & \large{Source coding} \\ \hline\hline
    Channel decoder & Source encoder \\ \hline
    Encoder 1 input & Decoder 1 input \\
    $(M_0,M_1)\in\{1,\dots,2^{n(R_0+R_1)}\}$ & $(T_0,T_1)\in\{1,\dots,2^{n(R_0+R_1)}\}$ \\ \hline
    Encoder 1 output $X_1\in\mathcal{X}_1$ & Decoder 1 output $\hat{X}_1\in\hat{\mathcal{X}}_1$ \\ \hline
    Encoder 2 input & Decoder 2 input \\
    $M_0\in\{1,\dots,2^{nR_0}\}$, & $T_0\in\{1,\dots,2^{nR_0}\}$, \\
    $Z_i(X_{1,i}),M_{12}(M_0,M_1)$ & $Z_i(\hat{X}_{1,i}),T_{12}(T_0,T_1)$ \\ \hline
    Encoder 2 output $X_{2}\in\mathcal{X}_2$ & Decoder 2 output $\hat{X}_{2}\in\hat{\mathcal{X}}_2$ \\ \hline
    Decoder input $Y\in\mathcal{Y}$ & Encoder input $X\in\mathcal{X}$ \\ \hline
    Decoder output & Encoder output \\
    $(\hat{M}_0,\hat{M}_1)\in\{1,\dots,2^{n(R_0+R_1)}\}$ & $(T_0,T_1)\in\{1,\dots,2^{n(R_0+R_1)}\}$ \\ \hline
    Encoding function $f_1:\mathcal{M}_0\times\mathcal{M}_1\mapsto\mathcal{X}_1^n$ & Decoding function $g_1:\mathcal{T}_0\times\mathcal{T}_1\mapsto\hat{\mathcal{X}}^n$ \\ \hline
    Causal cribbing encoding function & Causal cribbing decoding function \\
    $f_2:\mathcal{M}_0\times\mathcal{M}_{12}\times\mathcal{Z}^i\mapsto\mathcal{X}_{2,i}$ & $g:\mathcal{T}_0\times\mathcal{T}_{12}\times\mathcal{Z}^i\mapsto\hat{\mathcal{X}}_{2,i}$ \\ \hline
    Decoding function & Encoding function\\
    $g:\mathcal{Y}^n\mapsto\mathcal{M}_0\times\mathcal{M}_1$ & $f_0:\mathcal{X}^n\mapsto\mathcal{T}_0$, $f_1:\mathcal{X}^n\mapsto\mathcal{T}_1$ \\ \hline
    Auxiliary RV $U$ & Auxiliary RV $U$ \\ \hline
    Joint distribution $p(u,x_1,x_2,y)$ & Joint distribution $p(u,\hat{x}_1,\hat{x}_2,x)$ \\ \hline
    Constraint: $p(y|x_1,x_2)$ is fixed & Constraint: $p(x)$ is fixed \\ \hline
    \hline
  \end{tabular}
\end{center}
\end{table}Table \ref{table:principles} describes the principles of duality between channel coding and source coding.
We start by defining the channel coding problem and state its capacity region. We continue by solving its rate distortion dual, i.e., the SR problem with combined cooperation and partial cribbing at the decoder. We end this section by pointing out the dualities between these two settings and show how the corner points of the two regions are the same.
\vspace{-2mm}
\subsection{The MAC with a Common Message, a Private Message, and Combined Cooperation and Partial Cribbing}
Let us define the setting depicted in Fig. \ref{fig:dual_channel}.\vspace{-2mm}
\begin{figure}[!h]
\vspace{-5mm}
\begin{center}
\begin{psfrags}
    \psfragscanon
    \psfrag{A}[][][1]{Encoder 1}
    \psfrag{B}[][][1]{Encoder 2}
    \psfrag{C}[][][1]{$Z_{i}=g(X_{1,i})$}
    \psfrag{E}[][][1]{$P_{Y|X_1,X_2}$}
    \psfrag{F}[][][1]{Decoder}
    \psfrag{G}[r][][1]{$m_0 \in 2^{nR_0}$}
    \psfrag{H}[r][][1]{$m_1 \in 2^{nR_1}$}
    \psfrag{I}[][][1]{$X_{1,i}(m_0,m_1)$}
    \psfrag{J}[][][1]{$X_{2,i}(m_0,m_{12},Z^{i-1})$}
    \psfrag{K}[][][1]{$Y_i$}
    \psfrag{L}[][][1]{$\hat{m_0},\hat{m_1}$}
    \psfrag{M}[][][1]{$C_{12}$}
\includegraphics[scale = 0.8]{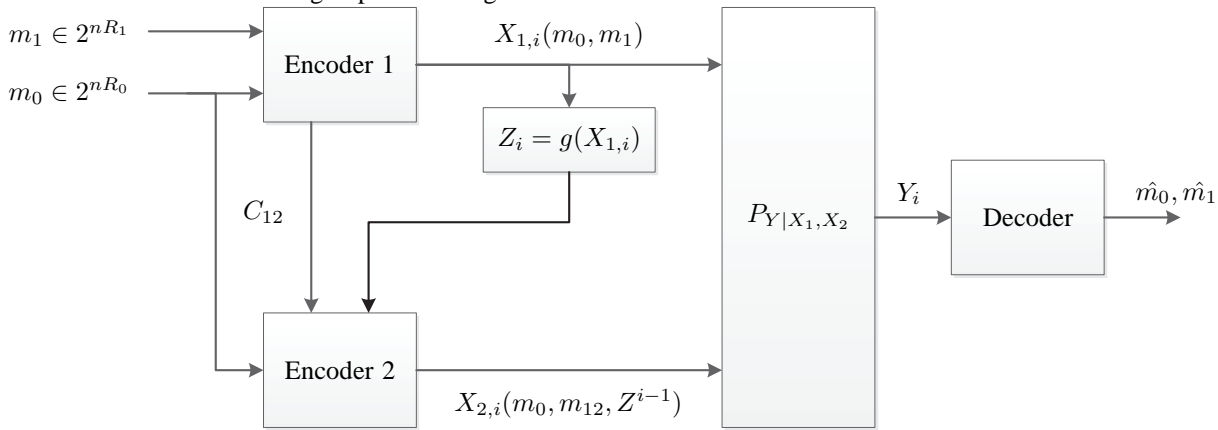}
\caption{MAC with common message, private message, and combined cooperation and cribbing. Encoder 2 obtains message $M_{12}$ prior to transmission. The cribbing is done causally.} \label{fig:dual_channel}
\psfragscanoff
\end{psfrags}
\end{center}
\vspace{-8mm}
\end{figure}
\begin{definition}\label{def:dual_channel}
A $(2^{nR_0},2^{nR_1},2^{nC_{12}},n)$ \textit{code} for the MAC with a common message, a private message, and combined cooperation and partial cribbing, as shown in Fig. \ref{fig:dual_channel}, consists at time $i$ of encoding functions at Encoder 1 and Encoder 2
\begin{eqnarray}
f_{12}&:&\{1,\dots,2^{nR_0}\}\times\{1,\dots,2^{nR_1}\}\mapsto\{1,\dots,2^{nC_{12}}\},\\
f_{1}&:&\{1,\dots,2^{nR_0}\}\times\{1,\dots,2^{nR_1}\} \mapsto  \mathcal{X}_{1}^n,\\
f_{2,i}&:&\{1,\dots,2^{nR_0}\}\times\{1,\dots,2^{nC_{12}}\}\times\mathcal{Z}^{i} \mapsto  \mathcal{X}_{2,i},
\end{eqnarray}
and a decoding function
\begin{eqnarray}
g:\mathcal{Y}^n \mapsto \{1,\dots,2^{nR_0}\}\times\{1,\dots,2^{nR_1}\}.
\end{eqnarray}
The average probability of error for a $(2^{nR_0},2^{nR_1},2^{nC_{12}},n)$ code is defined as
\begin{eqnarray}
P^{(n)}_e = \frac{1}{2^{n(R_0+R_1)}} \sum_{m_0,m_1} \Pr\{g(Y^n)\ne(m_0,m_1)|(m_0,m_1)\ \text{sent}\}.
\end{eqnarray}
\end{definition}
Let us define the following region and $\mathcal{R}_{MAC}$ that is contained in $ \mathbb{R}_+^2$, namely, contained in the set of nonnegative two-dimensional real numbers.
\begin{eqnarray}
\mathcal{R}_{MAC} =\left\{
                 \begin{array}{c}
                   R_1 \leq I(X_1; Y |Z, U)+ H(Z|U)+ C_{12},\\
                   R_0 + R_1 \leq I(X_1,U; Y ) $, for$\\
                   P(u)P(x_1|u)\mathbbm{1}_{z=f(x_1)}P(x_2|u, z)P(y|x_1, x_2).\\
                 \end{array}
               \right\}\label{eq:capacity_dual_channel}.
\end{eqnarray}
\begin{theorem}\label{theorem:dual_channel}{\emph{(Capacity Region of the MAC with Combined Cooperation and Partial Cribbing)}}
The capacity region of the MAC with common message, private message, and combined cooperation and causal partial cribbing, as described in Def. \ref{def:dual_channel}, is $\mathcal{R}_{MAC}$.
\end{theorem}
Since the proof for Theorem \ref{theorem:dual_channel} can be obtained by using the same methods described in Subsection \ref{sec:proof}, it is omitted for brevity. We go on to define the SR setting with combined cooperation and partial cribbing at the decoders.

\subsection{The Successive Refinement with Combined Cooperation and Partial Cribbing at the Decoders}
We address the rate distortion setting depicted in Fig. \ref{fig:dual_source}.
\begin{figure}[h]
\begin{psfrags}
    \psfragscanon
    \psfrag{A}[][][1]{Decoder 1}
    \psfrag{B}[][][1]{Decoder 2}
    \psfrag{C}[][][1]{$Z_{i}=g(\hat{X}_{1,i})$}
    \psfrag{F}[][][1]{Encoder}
    \psfrag{L}[][][1]{$T_0 \in 2^{nR_0}$}
    \psfrag{H}[][][1]{$T_1 \in 2^{nR_1}$}
    \psfrag{I}[l][][1]{$\hat{X}_{1}^n(T_0,T_1)$}
    \psfrag{J}[l][][1]{$\hat{X}_{2,i}(T_0,T_{12},Z^{i-1})$}
    \psfrag{K}[][][1]{$X^n$}
    \psfrag{M}[][][1]{$C_{12}$}
\includegraphics[scale = 0.8]{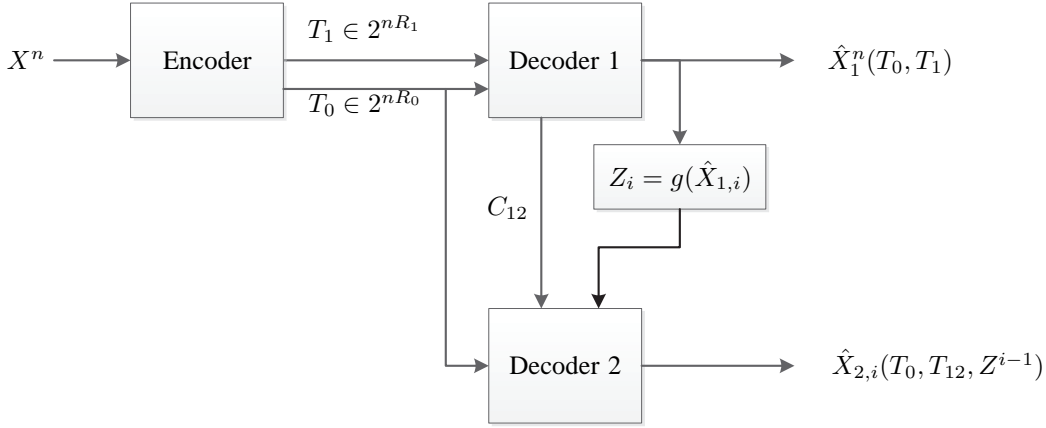}
\caption{SR with combined cooperation and partial cribbing at the decoders. The cribbing is done causally.} \label{fig:dual_source}
\psfragscanoff
\end{psfrags}
\end{figure}

The source sequence $X_i\in \mathcal{X}, i = 1, 2,\dots$ is drawn i.i.d. $\thicksim p(x)$. Let $\hat{\mathcal{X}}_1$ and $\hat{\mathcal{X}}_2$ denote the reconstruction alphabets, and $d_i: \mathcal{X}\times\hat{\mathcal{X}}_i\mapsto\left[0,\infty\right)$, for $i = 1, 2$ denote single letter distortion measures. Distortion between sequences is defined in the usual way;
\begin{equation}
    d_i(x^n,\hat{x}_i^n)=\frac{1}{n}\sum_{j=1}^nd_i(x_j,\hat{x}_{i,j}), \text{ for }i=1,2.
\end{equation}
\begin{definition}\label{def:dual_source}
A $(2^{nR_0},2^{nR_1},2^{nC_{12}},n)$ \textit{rate-distortion code} for the SR with combined cooperation and partial cribbing at the decoders, as shown in Fig. \ref{fig:dual_source}, consists at time $i$ of encoding functions at Encoder 1 and Encoder 2
\begin{eqnarray}
f_{0}&:& \mathcal{X}^n \mapsto \{1,\dots,2^{nR_0}\},\\
f_{1}&:& \mathcal{X}^n \mapsto \{1,\dots,2^{nR_1}\},\\
f_{12}&:&\{1,\dots,2^{nR_0}\}\times\{1,\dots,2^{nR_1}\} \mapsto \{1,\dots,2^{nC_{12}}\},\\
\end{eqnarray}
and a decoding function
\begin{eqnarray}
g_{1}&:& \{1,\dots,2^{nR_0}\}\times\{1,\dots,2^{nR_1}\} \mapsto \hat{\mathcal{X}}_{1}^n,\\
g_{2,i}&:&\{1,\dots,2^{nR_0}\}\times\{1,\dots,2^{nC_{12}}\}\times\mathcal{Z}^{i} \mapsto  \mathcal{X}_{2,i}.
\end{eqnarray}
\end{definition}
A rate $(R_0,R_1,D_1,D_2)$ is said to be \textit{achievable} for the SR with combined cooperation and partial cribbing at the decoders if $\forall \epsilon >0$ and a $(2^{nR_0},2^{nR_1},2^{nC_{12}},n)$ rate-distortion code the expected distortion for the decoders is bounded as,
\begin{equation}
    E\left[d_i(X^n,\hat{X}_i^n)\right]\leq D_i+\epsilon\text{, for $i=1,2$.}
\end{equation}
The \textit{rate-distortion region} $\mathcal{R}(D_1,D_2)$ is defined as the closure of the set of all achievable rate-distortion tuples $(R_0,R_1,D_1,D_2)$.

Let us define the following region $\mathcal{R}_{SR}(D_1,D_2)$ that is contained in $ \mathbb{R}_+^2$, namely, contained in the set of nonnegative two-dimensional real numbers.
\begin{eqnarray}
\mathcal{R}_{SR}(D_1,D_2) =\left\{
                 \begin{array}{c}
                   R_0 \geq I(X;Z, U)-H(Z|U)-C_{12},\\
                   R_0 + R_1 \geq I(\hat{X}_1,U; X )$, for$\\
                   P(x,x_1,u)\mathbbm{1}_{z=f(x_1),x_2=f(u,z_1)}$ s.t.$\\
                   E\left[d_i(X^n,\hat{X}_i^n)\right]\leq D_i + \epsilon$, for $i=1,2.
                 \end{array}
               \right\}\label{eq:capacity_dual_source}.
\end{eqnarray}
\begin{theorem}\label{theorem:dual_source}{\emph{(Rate Distortion Region of the Successive Refinement with Combined Cooperation and Partial Cribbing Decoders)}}
The rate-distortion region for the SR with combined cooperation and partial cribbing, as defined in Def. \ref{def:dual_source}, is $\mathcal{R}_{SR}(D_1,D_2)$.
\end{theorem}
\begin{proof}

\textit{Achievability:} The achievability for this model is the same as in \cite{asnani2013successive} where the achievable region was
\begin{eqnarray}
    \tilde{R}_0 &\geq& I(X;Z, U)-H(Z|U),\notag\\
    \tilde{R}_0 + \tilde{R}_1 &\geq& I(\hat{X}_1,U; X ).
\end{eqnarray}
In our case, we use rate splitting and set the following rates
\begin{eqnarray}
  \tilde{R}_0 &=& R_0 + C_{12},\\
  \tilde{R}_1 &=& R_0 - C_{12}.
\end{eqnarray}
By setting these rates we obtain the region in (\ref{eq:capacity_dual_source}).

\textit{Converse:} Assume we have a $(2^{nR_0},2^{nR_1},2^{nC_{12}},n)$ rate distortion code s.t. a $(R_0,R_1,D_1,D_2)$ tuple is feasible. For the first inequality
\begin{eqnarray}
  nR_0 &\geq& H(T_0) \\
    &\stackrel{(a)}=& H(Z^n,T_0,T_{12}) - H(Z^n|T_{12},T_0) - H(T_{12}|T_0) \\
    &\stackrel{(b)}\geq& I(X^n;Z^n,T_0,T_{12}) - H(Z^n|T_{12},T_0) - H(T_{12}) \\
    &\stackrel{(c)}\geq& \sum_{i=1}^n [I(X_i;Z^n,T_0,T_{12}|X^{i-1}) - H(Z_i|T_{12},T_0,Z^{i-1})] - nC_{12} \\
    &=& \sum_{i=1}^n [I(X_{i};Z^n,T_0,T_{12},X^{i-1}) - H(Z_i|T_{12},T_0,Z^{i-1})] - nC_{12}\\
    &\stackrel{(d)}\geq& \sum_{i=1}^n [I(X_{i};Z^i,T_0,T_{12}) - H(Z_i|T_{12},T_0,Z^{i-1})] - nC_{12}\\
    &\stackrel{(e)}=& \sum_{i=1}^n [I(X_{i};Z_i,U_i) - H(Z_i|U_i)] - nC_{12}\\
    &=& n \sum_{i=1}^n \frac{1}{n} [I(X_{i};Z_i,U_i) - H(Z_i|U_i)] - nC_{12}\\
    &\stackrel{(f)}=& n [I(X_{Q};Z_Q,U_Q|Q) - H(Z_Q,U_Q|Q) - C_{12}]\\
    &=& n [I(X_{Q};Z_Q,U_Q,Q) - H(Z_Q,U_Q|Q) - C_{12}]\\
    &\geq& n [I(X_{Q};Z_Q,U_Q) - H(Z_Q,U_Q) - C_{12}],
\end{eqnarray}
where (a) and (c) follow from the chain rule, (b) follows since conditionality reduces entropy, (d) follows since $X_{i}$ is independent of $X^{i-1}$, (e) follows by setting the random variable $U_i = (Z^{i-1},T_0,T_{12})$, and (f) follows by defining the RV $Q$ independent of $X^n$ and uniformly distributed over the set $\{1,2,3,\dots,n\}$.
For the second inequality
\begin{eqnarray}
  n(R_0 + R_1) &\geq& H(T_0,T_1) \\
    &\stackrel{(a)}=& I(X^n;T_0,T_1) \\
    &=& \sum_{i=1}^n I(X_i;T_0,T_1|X^{i-1})\\
    &\stackrel{(b)}=& \sum_{i=1}^n I(X_i;T_0,T_1,X^{i-1})\\
    &\stackrel{(c)}=& \sum_{i=1}^n I(X_i;T_0,T_1,\hat{X}_{1,i},Z^{i-1},T_{12},X^{i-1})\\
    &\geq& \sum_{i=1}^n I(X_i;\hat{X}_{1,i},Z^{i-1},T_0,T_{12})\\
    &=& \sum_{i=1}^n I(X_i;\hat{X}_{1,i},U_i)\\
    &=& n I(X_Q;\hat{X}_{1,Q},U_Q),
\end{eqnarray}
where (a) follows since $(T_0,T_1)$ is a function of $X^n$, (b) follows since since $X_{1,i}$ is independent of $X_1^{i-1}$, and (c) follows since $(\hat{X}_{1,i},Z^{i-1},T_{12})$ is a function of $(T_0,T_1)$. We complete the proof by noting that the joint distribution of $(X_Q, \hat{X}_{1,Q}, Z_Q, U_Q)$ is the same as that of $(X, \hat{X}_{1}, Z, U)$.
\end{proof}

\subsection{Duality Results Between the MAC and the Successive Refinement settings with combined cooperation and partial cribbing}
After establishing Theorems \ref{theorem:dual_channel} and \ref{theorem:dual_source}, we now point out the dualities between the two settings. The similarity between the rate regions of the two settings is evident.
Let us consider the corner points depicted in Table \ref{table-corner_points} and Fig. \ref{fig:dualityRegion}.
\begin{table}[!h]
\vspace{-4mm}
\begin{center}
  \caption{Corner points of MAC and SR}\label{table-corner_points}
  \begin{tabular}{| c | c |}
    \hline
    & $(R_0,R_1)$ \\ \hline
    MAC & $(I(Y;Z, U)-H(Z|U)-C_{12},I(Y;X_1|Z,U)+H(Z|U)+C_{12})$ \\
    (Theorem \ref{theorem:dual_channel}) & $(I(Y;X_1,U),0)$ \\ \hline
    SR & $(I(X;Z, U)-H(Z|U)-C_{12},I(X;\hat{X}_1|Z,U)+H(Z|U)+C_{12})$ \\
    (Theorem \ref{theorem:dual_source}) & $(I(X;\hat{X}_1,U),0)$ \\
    \hline
  \end{tabular}
\end{center}
\vspace{-4mm}
\end{table}
\begin{figure}[!h]
\vspace{-4mm}
\begin{center}
\begin{psfrags}
    \psfragscanon
    \psfrag{A}[r][][1]{$A$}
    \psfrag{B}[][][1]{$I(Y;Z, U)-H(Z|U)-C_{12}$}
    \psfrag{C}[][][1]{$I(Y;X_1,U)$}
    \psfrag{D}[r][][1]{$B$}
    \psfrag{E}[][][1]{$I(X;Z, U)-H(Z|U)-C_{12}$}
    \psfrag{F}[][][1]{$I(X;\hat{X}_1,U)$}
    \psfrag{G}[][][1]{$R_1$}
    \psfrag{H}[][][1]{$R_1$}
    \psfrag{I}[][][1]{$R_0$}
    \psfrag{J}[][][1]{$R_0$}
    \psfrag{T}[][][1]{MAC with combined}
    \psfrag{R}[][][1]{cooperation and cribbing}
    \psfrag{S}[][][1]{SR with combined}
    \psfrag{U}[][][1]{cooperation and cribbing}
    \includegraphics[scale = 1.1]{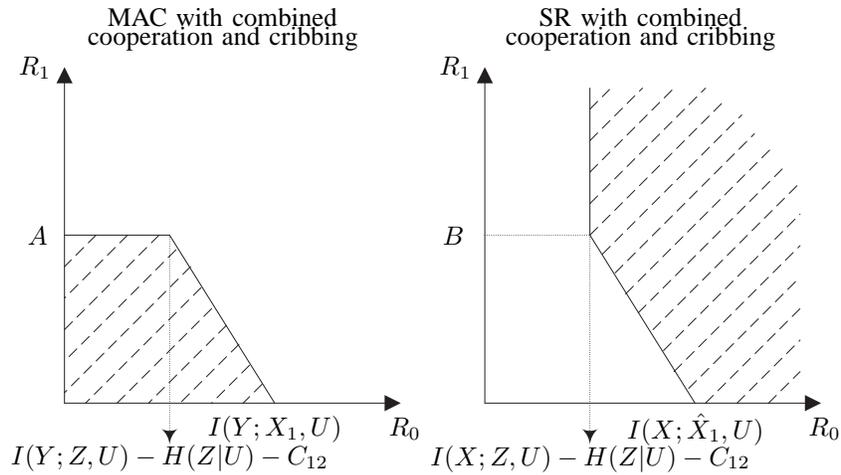}
\caption{Capacity region of the MAC and rate-distortion region of SR with combined cooperation and cribbing where $A$ is $I(Y;X_1|Z,U)+H(Z|U)+C_{12}$ and $B$ is $I(X;\hat{X}_1|Z,U)+H(Z|U)+C_{12}$.} \label{fig:dualityRegion}
\psfragscanoff
\end{psfrags}
\end{center}
\vspace{-4mm}
\end{figure}
One can see that the corner points are the same if we apply the duality rules $\hat{X}_1\leftrightarrow X_1$, $\hat{X}_2\leftrightarrow X_2$, $X\leftrightarrow Y$ and $\geq \leftrightarrow \leq$.
We notice that only one RV was used to describe the common message in both settings. This means that our methods of combining cooperation and cribbing can also be implemented in source coding problems. In the next section we address another case where only one RV is needed to describe both cooperation and cribbing.

\section{State-Dependent MAC with Combined Cooperation and Partial Cribbing}\label{sec:state}
Following our results from Section \ref{sec:def+results}, we now show that our methods can also be implemented for a state-dependent channel where still only one auxiliary RV is needed. Let us consider the MAC with cooperation and non-causal state known at a partially cribbing encoder and at the decoder, depicted in Fig. \ref{fig:coop+cribbing+decoder}.
\begin{figure}[h]
\begin{center}
\begin{psfrags}
    \psfragscanon
    \psfrag{A}[][][1]{Encoder 1}
    \psfrag{B}[][][1]{Encoder 2}
    \psfrag{D}[][][1]{$S^n$}
    \psfrag{E}[][][1]{$P_{Y|X_1,X_2,S}$}
    \psfrag{F}[][][1]{Decoder}
    \psfrag{G}[][][1]{$m_2 \in 2^{nR_2}$}
    \psfrag{H}[][][1]{$m_1 \in 2^{nR_1}$}
    \psfrag{I}[][][1]{$X_{1,i}(m_1,m_{21})$}
    \psfrag{J}[][][1]{$X_{2,i}(m_2,m_{12},Z^{i-1},S^n)$}
    \psfrag{K}[][][1]{$Y_i$}
    \psfrag{L}[][][1]{$(\hat{m_1},\hat{m_2})$}
    \psfrag{M}[][][1]{$C_{21}$}
    \psfrag{N}[][][1]{$C_{12}$}
    \psfrag{O}[][][1]{$Z_i=f(X_{1,i})$}
    \includegraphics[scale = 0.8]{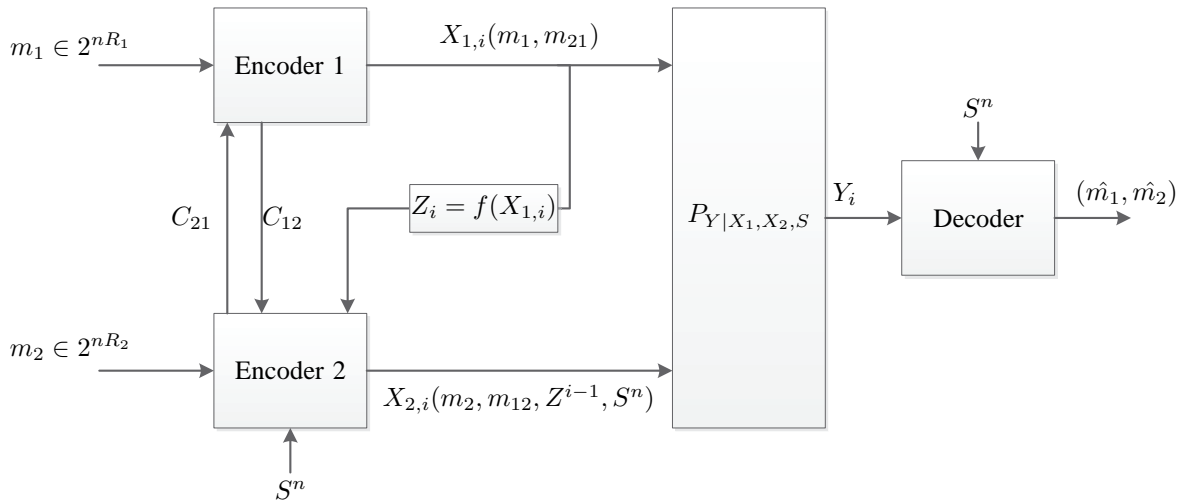}
\caption{The MAC with cooperation and state known at a partially cribbing encoder and at the decoder. Encoder 1 and Encoder 2 obtain messages $M_{21}$ and $M_{12}$ prior to transmission. The partial cribbing is done strictly causally only by Encoder 2. This setting corresponds to the strictly causal case.} \label{fig:coop+cribbing+decoder}
\psfragscanoff
\end{psfrags}
\end{center}
\end{figure}

We note that message $M_{12}$ is sent prior to message $M_{21}$.
For this model we address two different cases:
\begin{itemize}
  \item The strictly causal case (sc) : Encoder 2 obtains $Z_{i}$ with unit delay.
  \item The causal case (c) : Encoder 2 obtains $Z_{i}$ without delay.
\end{itemize}
The channel probability does not depend on the time index $i$ and is memoryless, i.e.,
\begin{equation}
P(y_i|x_1^i,x^i_2,s^i,y^{i-1}) = P(y_i|x_{1,i},x_{2,i},s_i)\label{eq:DMCS}
\end{equation}
\begin{definition}\label{coop+cribbing+decoder}
A $(2^{nR_1},2^{nR_2},2^{nC_{12}},2^{nC_{21}},n)$ \textit{code} for the MAC with cooperation and non-causal state known at a partially cribbing encoder and at the decoder, as shown in Fig. \ref{fig:coop+cribbing+decoder}, consists at time $i$ of encoding functions at Encoder 1 and Encoder 2.
\begin{eqnarray}
f_{12}&:&\{1,\dots,2^{nR_1}\} \mapsto  \{1,\dots,2^{nC_{12}}\}, \label{eq:en-relation-state1}\\ f_{21}&:&\{1,\dots,2^{nR_2}\}\times\mathcal{S}^{n}\times\{1,\dots,2^{nC_{12}}\} \mapsto  \{1,\dots,2^{nC_{21}}\}, \\
f_{1}&:&\{1,\dots,2^{nC_{21}}\}\times\{1,\dots,2^{nR_1}\} \mapsto  \mathcal{X}_{1}^n, \\ \label{eq:f1-state}
f^{sc}_{2,i} &:& \{1,\dots,2^{nC_{12}}\} \times \{1,\dots,2^{nR_2}\} \times \mathcal{S}^{n} \times \mathcal{Z}^{i-1} \mapsto  \mathcal{X}_{2,i}, \\
f^{c}_{2,i} &:&\{1,\dots,2^{nC_{12}}\}\times\{1,\dots,2^{nR_2}\}\times\mathcal{S}^{n}\times\mathcal{Z}^{i} \mapsto  \mathcal{X}_{2,i}, \label{eq:en-relation-state2}
\end{eqnarray}
and a decoding function
\begin{eqnarray}
g&:& \mathcal{S}^{n}\times\mathcal{Y}^n \mapsto \{1,\dots,2^{nR_1}\}\times\{1,\dots,2^{nR_2}\}.
\end{eqnarray}
The average probability of error for a $(2^{nR_1},2^{nR_2},2^{nC_{12}},2^{nC_{21}},n)$ code is defined as
\begin{eqnarray}
P^{(n)}_e = \frac{1}{2^{n(R_1+R_2)}} \sum_{m_1,m_2} \Pr\{g(Y^n,S^n)\ne(m_1,m_2)|(m_1,m_2)\ \text{sent}\}.
\end{eqnarray}
\end{definition}
Let us define the following regions, $\mathcal{R}^{sc}_{State}$ and $\mathcal{R}^{c}_{State}$, that are contained in $\mathbb{R}_+^2$, namely, contained in the set of nonnegative two-dimensional real numbers.
\begin{eqnarray}
\mathcal{R}^{sc}_{State} =\left\{
                 \begin{array}{c}
                   C_{21}\geq I(U;S),\\
                   R_1 \leq H(Z |U)+I(X_1;Y|S,U,X_2,Z)+C_{12},\\
                   R_2 \leq I(X_2; Y |X_1, S,U) + C_{21} - I(U;S),\\
                   R_1 + R_2 \leq I(X_1, X_2 ; Y|S),\\
                   R_1 + R_2 \leq I(X_1, X_2 ; Y|U, Z, S)+H(Z|U) + C_{12}  + C_{21} - I(U;S)$, for$\\
                   P(s)P(u|s)P(x_1|u)\mathbbm{1}_{z=f(x_1)}P(x_2|s,u)P(y|x_1,x_2,s).
                 \end{array}
               \right\}\label{eq:capacityS}
\end{eqnarray}
The region $\mathcal{R}^{c}_{State}$ is defined with the same set of inequalities as in (\ref{eq:capacityS}), but the joint distribution is of the form
\begin{eqnarray}
                   P(s)P(u|s)P(x_1|u)\mathbbm{1}_{z=f(x_1)}P(x_2|s,u,z)P(y|x_1,x_2,s)\label{eq:state-joint-B}.
\end{eqnarray}
\begin{theorem}\label{theorem:add2}{\emph{(Capacity Region of the MAC with Cooperation and State Known at a Partial Cribbing Encoder)}}
The capacity regions of the MAC with cooperation and non-causal state known at a partially cribbing encoder and at the decoder for the strictly causal case and the causal case, as described in Def. \ref{coop+cribbing+decoder}, are $\mathcal{R}^{sc}_{State}$ and $\mathcal{R}^{c}_{State}$, respectively.
\end{theorem}
The role of the RV $U$ is to generate an empirical coordination between the two encoders regarding the state channel and to generate a common message between the two encoders by combining the cooperation links and the partial cribbing. We now examine two special cases of this capacity region.

\textit{Case 1: The One-Sided Cooperation and No Cribbing Case, i.e., $|\mathcal{Z}|=1$ and $C_{12}=0$:} In this case $H(Z|U)= 0$ and hence the region $\mathcal{R}^{sc}_{State}$ coincides with the region in \cite[Theorem 1]{permuter2011message}.

\textit{Case 2: $|\mathcal{S}|=1$, The Memoryless Case:} Notice that in this case $I(U;S)=0$ and the region $\mathcal{R}^{sc}_{State}$ reduces to
\begin{eqnarray}
\mathcal{R}^2_{State} =\left\{
             \begin{array}{c}
               R_1 \leq H(Z |U)+I(X_1;Y|U,X_2,Z)+C_{12},\\
               R_2 \leq I(X_2; Y |X_1, U) + C_{21},\\
               R_1 + R_2 \leq I(X_1, X_2 ; Y),\\
               R_1 + R_2 \leq I(X_1, X_2 ; Y|U, Z)+H(Z|U) + C_{12}  + C_{21}$, for$\\
               P(u)P(x_1|u)\mathbbm{1}_{z=f(x_1)}P(x_2|u)P(y|x_1,x_2).
             \end{array}
           \right\}\label{eq:capacityS2}
\end{eqnarray}
which is the region in Theorem \ref{theorem:MAC+crib+coop} where $Z_1=Z$ and only Encoder 2 cribs from Encoder 1, i.e., $|\mathcal{Z}_2|=1$.

The proof of Theorem \ref{theorem:add2} is given in Appendix \ref{appendix:crib+coop+state+decoder}.

Although we have shown that for combined cooperation and cribbing only one auxiliary RV is needed to describe the capacity region, in some cases this is not possible.
For instance, if the role of the cribbing and cooperation in the communication setting is different, then more then one auxiliary RV is needed. In the next section, we introduce a MAC with cooperation and action-dependent state known at a cribbing encoder. Because of the nature of actions and of non-causal states, the actions depend only on the cooperation and, therefore, two auxiliary RVs are needed, one for the cooperation and one for the cribbing.

\section{MAC with Cooperation and Action-dependent State known at a Cribbing Encoder}\label{sec:action}
We now address a MAC where two auxiliary RVs are needed in order to combine cooperation and cribbing. Consider the MAC with one-way cooperation and action-dependent state known at a cribbing encoder, depicted in Fig. \ref{fig:action+cribbing+coop}. Notice that the action $A^n$ is taken from $(m_2,m_{12})$.
\begin{figure}[h]
\begin{center}
\begin{psfrags}
    \psfragscanon
    \psfrag{A}[][][1]{Encoder 1}
    \psfrag{B}[][][1]{Encoder 2}
    \psfrag{C}[][][1]{$A^n$}
    \psfrag{D}[][][1]{$S^n$}
    \psfrag{E}[][][1]{$P_{Y|X_1,X_2,S}$}
    \psfrag{F}[][][1]{Decoder}
    \psfrag{G}[][][1]{$m_2 \in \{1,\dots,2^{nR_2}\}$}
    \psfrag{H}[][][1]{$m_1 \in \{1,\dots,2^{nR_1}\}$}
    \psfrag{I}[][][1]{$X_{1,i}(m_1)$}
    \psfrag{J}[][][1]{$X_{2,i}(m_2,m_{12},X_1^{i-1},S^n)$}
    \psfrag{K}[][][1]{$Y_i$}
    \psfrag{L}[][][1]{$(\hat{m}_1,\hat{m}_2)$}
    \psfrag{O}[][][1]{$p(s|a)$}
    \psfrag{M}[][][1]{$C_{12}$}
\includegraphics[scale = 0.9]{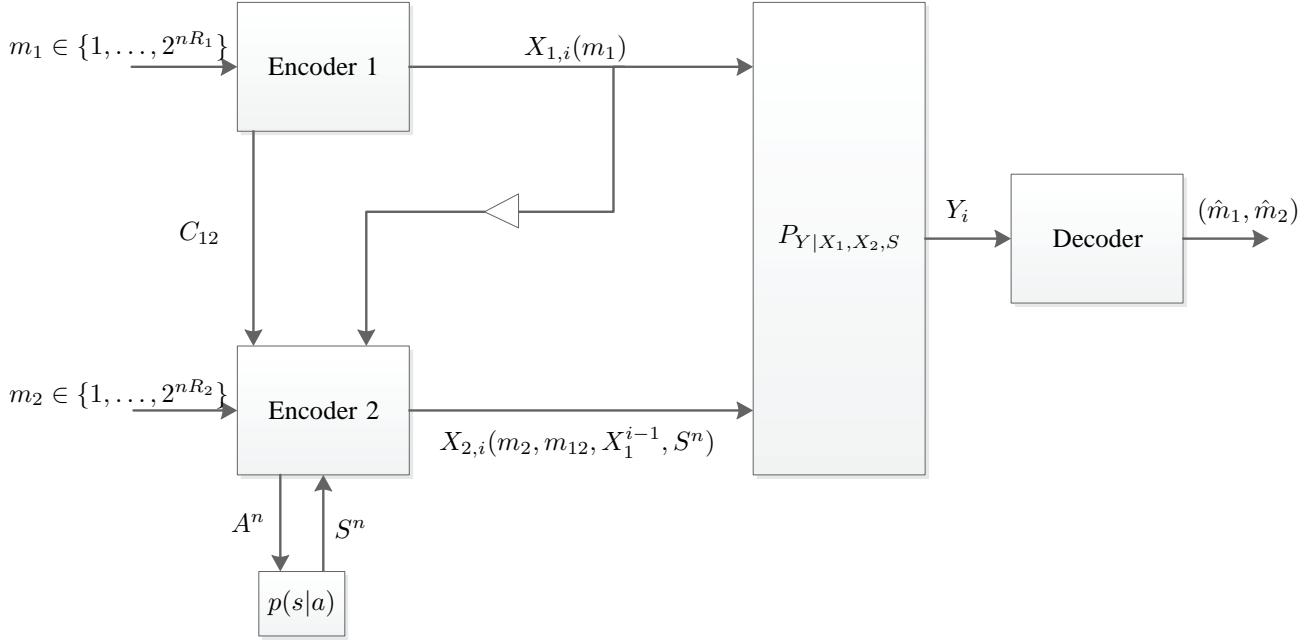}
\caption{The MAC with one-way cooperation and action-dependent state known at a cribbing encoder. Encoder 2 obtains messages $M_{12}$ prior to transmission. The cribbing is done strictly causally only by Encoder 2. This setting corresponds to the strictly causal case.} \label{fig:action+cribbing+coop}
\psfragscanoff
\end{psfrags}
\end{center}
\end{figure}

We address two cases for this setting:
\begin{itemize}
  \item The strictly causal case (sc) : Encoder 2 obtains $X_{1,i}$ with unit delay.
  \item The causal case (c) : Encoder 2 obtains $X_{1,i}$ without delay.
\end{itemize}
The channel probability is defined as in (\ref{eq:DMCS}).
\begin{definition}\label{action+cribbing}
A $(2^{nR_1},2^{nR_2},2^{nC_{12}},n)$ \textit{code} for the MAC with one-way cooperation and action-dependent state known at a cribbing encoder, as shown in Fig. \ref{fig:action+cribbing+coop}, consists at time $i$ of encoding functions at Encoder 1 and Encoder 2
\begin{eqnarray}
f_{12}&:&\{1,\dots,2^{nR_1}\} \mapsto  \{1,\dots,2^{nC_{12}}\}, \\
f_{1}&:&\{1,\dots,2^{nR_1}\} \mapsto  \mathcal{X}_{1}^n, \label{eq:f1} \\
f_{Action}&:&\{1,\dots,2^{nR_2}\}\times\{1,\dots,2^{nC_{12}}\} \mapsto \mathcal{A}^n, \label{eq:f3}\\
f^{sc}_{2,i}&:&\{1,\dots,2^{nR_2}\}\times\{1,\dots,2^{nC_{12}}\}\times\mathcal{S}^{n}\times\mathcal{X}^{i-1}_1 \mapsto  \mathcal{X}_{2,i},\\
f^c_{2,i}&:&\{1,\dots,2^{nR_2}\}\times\{1,\dots,2^{nC_{12}}\}\times\mathcal{S}^{n}\times\mathcal{X}^{i}_1 \mapsto  \mathcal{X}_{2,i},
\end{eqnarray}
and a decoding function
\begin{eqnarray}
g&:& \mathcal{Y}^n \mapsto \{1,\dots,2^{nR_1}\}\times\{1,\dots,2^{nR_2}\}.
\end{eqnarray}
The average probability of error for a $(2^{nR_1},2^{nR_2},2^{nC_{12}},n)$ code is defined as
\begin{eqnarray}
P^{(n)}_e = \frac{1}{2^{n(R_1+R_2)}} \sum_{m_1,m_2} \Pr\{g(Y^n)\ne(m_1,m_2)|(m_1,m_2)\ \text{sent}\}.
\end{eqnarray}
\end{definition}
Let us define the following regions $\mathcal{R}^{sc}_{Action}$ and $\mathcal{R}^c_{Action}$ that are contained in $\mathbb{R}_+^2$, namely, contained in the set of nonnegative two-dimensional real numbers.
\begin{eqnarray}
\mathcal{R}^{sc}_{Action} =\left\{
                 \begin{array}{c}
                   R_1 \leq \min\{H(X_1 |V,W),I(Y ;V ,X_1 ,U |W,A)-I(S;U|W,V,A)\} + C_{12},\\
                   R_2 \leq I(U, A ; Y |X_1, V, W) - I(U;S|W,V,A),\\
                   R_1 + R_2 \leq I(X_1, V, U, A ; Y|W) - I(U;S|W,V,A) + C_{12},\\
                   R_1 + R_2 \leq I(X_1, V, U, A, W; Y) - I(U;S|W,V,A)$, for$\\
                   P(w)P(v|w)p(a|w)P(s|a)P(x_1|v,w)P(u,x_2|s,v,a,w)P(y|x_1,x_2,s).
                 \end{array}
               \right\}\label{eq:capacity8}
\end{eqnarray}
The region $\mathcal{R}^c_{Action}$ is defined with the same set of inequalities as in (\ref{eq:capacity8}), but the joint distribution is of the form
\begin{eqnarray}
                   P(w)P(v|w)p(a|w)P(s|a)P(x_1|v,w)P(u|s,v,a,w)P(x_2|v,u,s,a,w,x_1)P(y|x_1,x_2,s).\label{eq:action-joint-B}
\end{eqnarray}
\begin{theorem}\label{theorem:add1}{\emph{(Capacity Region of the MAC with Cooperation and Action-Dependent State Known at a Cribbing Encoder)}}
The capacity regions of the MAC with one-way cooperation and action-dependent state known at a strictly causal and causal cribbing encoder, as described in Def. \ref{coop+cribbing+decoder}, are $\mathcal{R}^{sc}_{Action}$ and $\mathcal{R}^c_{Action}$, respectively.
\end{theorem}

In this case, $U$ is a Gelfand-Pinsker coding RV \cite{gel1980coding}. The role of the RV $W$ is to generate a common message based on the cooperation link, whereas the RV $V$ generates a common message based on the cribbing. The reason why in this case we cannot combine the cooperation and cribbing is that only part of the common information of both encoders is being used to generate the action sequence $A^n$. This example shows that in cases where only part of the common information that the encoders share is being used for arbitrary purposes, cooperation and cribbing cannot be combined into one RV. We now address two previous results in this field and show that they are special cases of our result.

\textit{Case 1: The Action-Dependent MAC where $C_{12}=R_1$:} In this case the region reduces to
\begin{eqnarray}
\mathcal{R}^1_{Action} =\left\{
             \begin{array}{c}
               R_2 \leq I(U, A ; Y |X_1, V, W) - I(U;S|W,V,A),\\
               R_1 + R_2 \leq I(X_1, V, U, A, W ; Y) - I(U;S|W,V,A)$, for$\\
               P(w)P(v|w)p(a|w)P(s|a)P(x_1|v,w)P(u,x_2|s,v,a,w)P(y|x_1,x_2,s).
             \end{array}
           \right\}\label{eq:capacityA1}
\end{eqnarray}
First, we notice that the cribbing in this case is redundant. Second, since the action is now taken from $(M_1,M_2)$ we can set the RV $W=X_1$ and $V$ as a constant and the region coincides with the capacity region in \cite{dikstein2012mac}.

\textit{Case 2: The State-Dependent MAC with State Known at a Cribbing Encoder, i.e., $|\mathcal{A}|=1$ and $C_{12}=0$:} Notice that in this case the state is not action-dependent and the region reduces to
\begin{eqnarray}
\mathcal{R}^2_{Action} =\left\{
             \begin{array}{c}
               R_1 \leq H(X_1 |V,W),\\
               R_2 \leq I(U ; Y |X_1, V, W) - I(U;S|W,V),\\
               R_1 + R_2 \leq I(X_1, V, U ; Y|W) - I(U;S|W,V),\\
               R_1 + R_2 \leq I(X_1, V, U, W ; Y) - I(U;S|W,V)$, for$\\
               P(w)P(v|w)P(s)P(x_1|v,w)P(u,x_2|s,v,w)P(y|x_1,x_2,s).
             \end{array}
           \right\}\label{eq:capacityA2}
\end{eqnarray}
If we set $W$ as constant, the region coincides with the capacity region in \cite{bross2010state}. Since these regions are equal, this shows that the capacity region in \cite{bross2010state} is a special case of the region in Theorem \ref{theorem:add1}.

The proof of Theorem \ref{theorem:add1} is given in Appendix \ref{appendix:crib+coop+state}.

\section{Conclusions and Future Work}\label{sec:conc}
In this paper, we have presented the capacity region for the MAC with combined cooperation and partial cribbing. Remarkably, the solution necessitates the use of only one auxiliary RV. Additionally, we have shown an achievability scheme for the Gaussian MAC with combined one-sided cooperation and causal partial cribbing. In this case, partial cribbing is a scalar quantization of Encoder 1's output obtained by Encoder 2. Graphs of achievability regions were presented for various number of quantization bits and capacity links. Using these results, it is possible to find under which conditions the outer bound is achieved. Thereafter, we considered a dual setting for the MAC with a common message, a private message, and combined cooperation and cribbing. We successfully characterized the rate-distortion region for the dual model using a single auxiliary RV. We applied our methods in order to find the capacity region for a MAC with cooperation and state known at a cribbing encoder and at the decoder. Again, the capacity region consisted of only one auxiliary RV. Finally, we addressed a MAC with one-way cooperation and cribbing and action-dependent state, where the action was based on the cooperation between the encoders. In this case two auxiliary RVs were needed. We stated that if only part of the common information that the encoders share is being used for arbitrary purposes, then cooperation and cribbing cannot be combined into one RV. We suggest, for future work, considering the non-causal partial cribbing case and the interference channel with combined cooperation and cribbing. An additional case to consider is where the state or action is known at the weak encoder (the non-cognitive encoder).

\appendices
\section{Achievability for the MAC with a Common Message and Partially Cribbing Encoders}\label{appendix:achieve}
Fix a joint distribution $P(u)P(x_1|u)\mathbbm{1}_{z_1=f(x_1)}P(x_2|u)\mathbbm{1}_{z_2=f(x_2)}P(y|x_1, x_2)$. In the following achievability scheme we use Block Markov Coding and Rate-Splitting.

\emph{\underline{Coding Scheme:}} We consider $B$ blocks, each consisting of $n$ symbols and thus we transmit $nB$ symbols. We transmit $B-1$ message-pairs $(M_1,M_2)$ in $B$ blocks of information. Here, $M_i \in \{1,\dots,2^{nR_i}\}$ for $i \in \{1,2\}$; thus, asymptotically, for a large enough $n$, our transmission rate would be $\frac{nR_i(B-1)}{nB}\stackrel{n\rightarrow\infty}\longrightarrow R_i$ for $i \in \{1,2\}$. In each block we split messages $M_{1}$ and $M_{2}$ into $(M^\prime_{1},M^{\prime\prime}_{1})$ and $(M^\prime_{2},M^{\prime\prime}_{2})$, respectively, s.t. $R_1 = R_1^\prime + R_1^{\prime\prime}$ and $R_2 = R_2^\prime + R_2^{\prime\prime}$.

\emph{\underline{Code Design:}} Generate $2^{n(R_0+R_1^\prime+R_2^\prime)}$ codewords $u^n$ i.i.d. using $P(u^n)=\Pi_{i=1}^nP(u_i)$. For each $u^n$, generate $2^{nR_1^\prime}$ codewords $z_1^n$ i.i.d. using $P(z_1^n|u^n)=\Pi_{i=1}^nP(z_{1,i}|u_i)$ and $2^{nR_1^{\prime\prime}}$ codewords $x_1^n$ i.i.d. using $P(x_1^n|u^n,z_1^n)=\Pi_{i=1}^nP(x_{1,i}|u_i,z_{1,i})$. Additionally, for each $u^n$, generate $2^{nR_2^\prime}$ codewords $z_2^n$ i.i.d. using $P(z_2^n|u^n)=\Pi_{i=1}^nP(z_{2,i}|u_i)$ and $2^{nR_2^{\prime\prime}}$ codewords $x_2^n$ i.i.d. using $P(x_2^n|u^n,z_2^n)=\Pi_{i=1}^nP(x_{2,i}|u_i,z_{2,i})$.

\emph{\underline{Encoding:}} We denote the realizations of the sequences $(M_0,M_1,M_2)$ at block $b$ as $(m_{0,b},m_{1,b},m_{2,b})$.
Since we use block Markov coding, we set $(m^\prime_{1,B},m^\prime_{1,B})=(1,1)$.
In block $b \in \{1,\dots,B\}$, encode message $(m_{0,b},m^\prime_{1,b-1},m^\prime_{2,b-1})$ using $u^n(m_{0,b},m^\prime_{1,b-1},m^\prime_{2,b-1})$.
Encode message $m^\prime_{1,b}$ conditioned on $(m_{0,b},m^\prime_{1,b-1},m^\prime_{2,b-1})$ using $z_1^n(m^\prime_{1,b},u^n)$ and message $m^{\prime\prime}_{1,b}$ conditioned on $(m_{0,b},m^\prime_{1,b-1},m^\prime_{2,b-1},m^\prime_{1,b})$ using $x_1^n(m^{\prime\prime}_{1,b},u^n,z_1^n)$.
Additionally, encode message $m^\prime_{2,b}$ conditioned on $(m_{0,b},m^\prime_{1,b-1},m^\prime_{2,b-1})$ using $z_2^n(m^\prime_{2,b},u^n)$ and message $m^{\prime\prime}_{2,b}$ conditioned on $(m_{0,b},m^\prime_{1,b-1},m^\prime_{2,b-1},m^\prime_{2,b})$ using $x_2^n(m^{\prime\prime}_{2,b},u^n,z_2^n)$.
Send $x_1^n(m^{\prime\prime}_{1,b},u^n,z_1^n)$ and $x_2^n(m^{\prime\prime}_{2,b},u^n,z_2^n)$ over the channel.

\emph{\underline{Decoding at Encoder 1:}} At the end of block $b$, Encoder 1 tries to decode message $m^\prime_{2,b}$. Given $(m_{0,b},m^\prime_{1,b-1})$ and assuming that message $m^\prime_{2,b-1}$ was decoded correctly at the end of block $b-1$, Encoder 1 looks for $\hat{m}^\prime_{2,b}$ s.t.
\begin{equation}
    (u^n(m_{0,b},m^\prime_{1,b-1},m^\prime_{2,b-1}),z_2^n(\hat{m}^\prime_{2,b},u^n))\in T_\epsilon^{(n)}(U,Z_2).
\end{equation}
If no such $\hat{m}^\prime_{2,b}$, or more than one such $\hat{m}^\prime_{2,b}$, was found, an error is declared at block $b$ and therefore in the whole super-block $nB$.

\emph{\underline{Decoding at Encoder 2:}} Similarly for Encoder 2; at the end of block b, Encoder 2 tries to decode message $m^\prime_{1,b}$. Given $(m_{0,b},m^\prime_{2,b-1})$ and assuming that message $m^\prime_{1,b-1}$ was decoded correctly at the end of block $b-1$, Encoder 2 looks for $\hat{m}^\prime_{1,b}$ s.t.
\begin{equation}
    (u^n(m_{0,b},m^\prime_{1,b-1},m^\prime_{2,b-1}),z_1^n(\hat{m}^\prime_{1,b},u^n))\in T_\epsilon^{(n)}(U,Z_1).
\end{equation}
If no such $\hat{m}^\prime_{1,b}$, or more than one such $\hat{m}^\prime_{1,b}$, was found, an error is declared at block $b$ and therefore in the whole super-block $nB$.

\emph{\underline{Decoding at the receiver:}} At the end of block $B$, the decoding is done backwards. At block $b$, the decoder looks for the triplet $(\hat{m}_{0,b},\hat{m}^\prime_{1,b-1},\hat{m}^{\prime\prime}_{1,b},\hat{m}^\prime_{2,b-1},\hat{m}^{\prime\prime}_{2,b})$ s.t.
\begin{eqnarray}
    (u^n(\hat{m}_{0,b},\hat{m}^\prime_{1,b-1},\hat{m}^\prime_{2,b-1}),z_1^n(\hat{m}^\prime_{1,b},u^n),z_2^n(\hat{m}^\prime_{2,b},u^n),x_1^n(\hat{m}^{\prime\prime}_{1,b},u^n,z_1^n),x_2^n(\hat{m}^{\prime\prime}_{2,b},u^n,z_2^n),y^n)\notag\\
    \tab\in T_\epsilon^{(n)}(U,Z_1,Z_2,X_1,X_2,Y).
\end{eqnarray}
If no such tuple, or more than one such tuple, was found, an error is declared at block $b$ and therefore in the whole super-block $nB$.

\emph{\underline{Error Analysis:}} The probability that $z_1^n(1,u^n)=z_1^n(i,u^n)$ where $i>1$ and where $(u^n(i),z_1^n(1,u^n))\in T_\epsilon^{(n)}(U,Z_1)$ is bounded by $2^{-n(H(Z_1|U)-\delta(\epsilon))}$, where $\delta(\epsilon)$ goes to zero as $\epsilon$ goes to zero. Hence, if
\begin{equation}\label{eq:E1}
    R^\prime_1<H(Z_1|U),
\end{equation}
then the probability that an incorrect message $m^\prime_{1,b}$ was decoded goes to zero for a large enough $n$.

From symmetry, we can see that if
\begin{equation}\label{eq:E2}
    R^\prime_2<H(Z_2|U),
\end{equation}
then the probability that an incorrect message $m^\prime_{2,b}$ was decoded goes to zero for a large enough $n$.
We define the following event at block $b$:
\begin{eqnarray}
    E_{i,j,k,b} \triangleq (u^n(i),z_1^n(\hat{m}^\prime_{1,b},u^n),z_2^n(\hat{m}^\prime_{2,b},u^n),x_1^n(j,u^n,z_1^n),x_2^n(k,u^n,z_2^n),y^n)\in T_\epsilon^{(n)}(U,Z_1,Z_2,X_1,X_2,Y).
\end{eqnarray}
We can bound the probability of error as follows:
\begin{eqnarray}
    P_{e,b}^{(n)} \leq \Pr(E^c_{1,1,1,b}) &+& \sum_{i=1,j>1,k=1}\Pr(E_{1,j,1,b}) + \sum_{i=1,j=1,k>1}\Pr(E_{1,1,k,b})\notag\\ &+& \sum_{i=1,j>1,k>1}\Pr(E_{1,j,k,b}) + \sum_{i>1,j>1,k>1}\Pr(E_{i,j,k,b}). \label{eq:prob_of_error}
\end{eqnarray}
We now show that each term in (\ref{eq:prob_of_error}) goes to zero for a large enough $n$.
\begin{itemize}
  \item Upper-bounding $\Pr(E^c_{1,1,1,b})$: Since we assume that Transmitters 1 and 2 encode the correct message triplet $(m_{0,b},m_{1,b-1},m_{2,b-1})$ at block $b$ and that the receiver decoded the right triplet $(m_{0,b+1},m_{1,b},m_{2,b})$ at block $b+1$, by the law of large numbers (LLN), $\Pr(E^c_{1,1,1,b}) \rightarrow 0$ when $n\rightarrow\infty$.
  \item Upper-bounding $\sum_{i=1,j>1,k=1}\Pr(E_{1,j,1,b})$: Assuming that $(m^\prime_{1,b},m^\prime_{2,b})$ were decoded correctly at block $b+1$, the probability for this event is bounded by
      \begin{equation}
            \sum_{i=1,j>1,k=1}\Pr(E_{1,j,1,b}) \leq 2^{nR_1^{\prime\prime} }2^{n(I(X_1;Y|U,Z_1,X_2)-\delta(\epsilon)}. \label{eq:E3}
      \end{equation}
  \item Upper-bounding $\sum_{i=1,j=1,k>1}\Pr(E_{1,1,k,b})$: From symmetry,
      \begin{equation}
            \sum_{i=1,j=1,k>1}\Pr(E_{1,1,k,b}) \leq 2^{nR_2^{\prime\prime}}2^{n(I(X_2;Y|U,Z_2,X_1)-\delta(\epsilon)}. \label{eq:E4}
      \end{equation}
  \item Upper-bounding $\sum_{i=1,j>1,k>1}\Pr(E_{1,j,k,b})$: Again we assume that $(m^\prime_{1,b},m^\prime_{2,b})$ were decoded correctly at block $b+1$; the probability for this event is bounded by
      \begin{equation}
            \sum_{i=1,j>1,k>1}\Pr(E_{1,j,k,b}) \leq 2^{n(R_1^{\prime\prime} + R_2^{\prime\prime})}2^{n(I(X_1,X_2;Y|U,Z_1,Z_2)-\delta(\epsilon)}. \label{eq:E5}
      \end{equation}
  \item Upper-bounding $\sum_{i>1,j>1,k>1}\Pr(E_{i,j,k,b})$: We assume that $(m^\prime_{1,b},m^\prime_{2,b})$ were decoded correctly at block $b+1$; the probability for this event is bounded by
      \begin{equation}
            \sum_{i>1,j>1,k>1}\Pr(E_{i,j,k,b}) \leq 2^{n(R_0+R_1+R_2)}2^{n(I(X_1,X_2;Y)-\delta(\epsilon)}. \label{eq:E6}
      \end{equation}
\end{itemize}
Using the Fourier-Motzkin Elimination on equations (\ref{eq:E1}), (\ref{eq:E2}), (\ref{eq:E3}), (\ref{eq:E4}), (\ref{eq:E5}), and (\ref{eq:E6}) yields the achievable region in (\ref{eq:capacity-mac-ce}), thus completing the proof.
\hfill $\blacksquare$

\section{Proof of Theorem \ref{theorem:add2}} \label{appendix:crib+coop+state+decoder}
\subsection{Converse}
\textit{Converse for the strictly causal case:}
Given an achievable rate-pair $(R_1,R_2)$ we need to show that there exists a joint distribution of the form $P(s)P(u|s)P(v|u)P(z,x_1|v,u)P(x_2|s,v,u)P(y|x_1,x_2,s)$ such that the inequalities in (\ref{eq:capacityS}) are satisfied. Since $(R_1,R_2)$ is an achievable rate-pair, there exists a $(2^{nR_1},2^{nR_2},2^{nC_{12}},2^{nC_{21}},n)$ code with an arbitrarily small error probability $P^{(n)}_e$. By Fano's inequality,
\begin{eqnarray}
H(M_1,M_2|Y^n,S^n)\leq n(R_1 + R_2)P^{(n)}_e + H(P^{(n)}_e).
\end{eqnarray}
We set
\begin{equation}
(R_1 + R_2)P^{(n)}_e + \frac{1}{n}H(P^{(n)}_e)\triangleq \epsilon_n,
\end{equation}
where $\epsilon_n \rightarrow 0$ as $P^{(n)}_e\rightarrow 0$. Hence,
\begin{eqnarray}
H(M_1|Y^n,M_2,S^n) \leq H(M_1,M_2|Y^n,S^n) \leq n\epsilon_n,\\
H(M_2|Y^n,M_1,S^n) \leq H(M_1,M_2|Y^n,S^n) \leq n\epsilon_n.
\end{eqnarray}
For $R_1$ we have the following:
\begin{eqnarray}
    nR_1 &=& H(M_1)\\
        &=& H(M_1|M_{12})+H(M_{12})\\
        &\stackrel{(a)}=&   H(M_1|M_{12},M_2,S^n)+H(M_{12})\\
        &=&                 I(M_1;Y^n|M_{12},M_2,S^n) + H(M_1|Y^n,M_{12},M_2,S^n)+H(M_{12})\\
        &\stackrel{(b)}\leq&I(M_1;Y^n|M_{12},M_2,S^n)+nC_{12} + n\epsilon_n\\
        &\stackrel{(c)}=&   I(X^n_1,Z^n;Y^n|M_{12},M_2,S^n)+nC_{12} + n\epsilon_n\\
        &\stackrel{(d)}=&   I(Z^n;Y^n|M_{12},M_2,S^n) + I(X^n_1;Y^n|M_{12},M_2,S^n,Z^n) +nC_{12} + n\epsilon_n\\
        &\stackrel{(e)}=&   \sum_{i=1}^n [I(Z_i;Y^n|M_{12},M_{21},M_2,Z^{i-1},S^n) + I(X^n_1;Y_i|Y^{i-1},M_{12},M_{21},M_2,S^n,Z^n)]\notag\\
        &&\tab + nC_{12} + n\epsilon_n\\
        &\stackrel{(f)}\leq&\sum_{i=1}^n [H(Z_i|M_{21},Z^{i-1},M_{12},S^{i-1}) + I(X^n_1;Y_i|Y^{i-1},M_{12},M_{21},M_2,S^n,Z^n,X_{2,i})]\notag\\ &&\tab + nC_{12} + n\epsilon_n \\
        &\stackrel{(g)}\leq&\sum_{i=1}^n [H(Z_i|M_{21},Z^{i-1},M_{12},S^{i-1}) + I(X_{1,i};Y_i|M_{21},S^i,Z^{i-1},M_{12},X_{2,i},Z_i)] \notag\\ &&\tab +  nC_{12} + n\epsilon_n \label{markov-YX1X2S}\\
        &\stackrel{(h)}=&\sum_{i=1}^n [H(Z_i|U_i) + I(X_{1,i};Y_i|U_i,X_{2,i},S_i,Z_i)] + nC_{12} + n\epsilon_n,
\end{eqnarray}
where (a) follows from the fact that the messages $M_1$ and $(M_2,S^n)$ are independent, (b) follows from Fano's inequality, (c) follows from the Markov chain $M_1 - (X^n_1,Z^n,M_{12},M_2,S^n) - Y^n$, (d) and (e) follow from the chain rule and since $M_{21}=f(S^n,M_2,M_{12})$, (f) follows since conditioning reduces entropy and since $X_{2,i}=f(S^n,Z^{i-1},M_{12},M_2)$, (g) follows from the Markov Chain $Y_i - (X_{1,i},X_{2,i},S^i,M_{12},M_{21},Z^i) - (Y^{i-1},M_2,S_{i+1}^n,Z_{i+1}^n)$, and (h) follows by setting the RV
\begin{eqnarray}
    U_i&\triangleq& (M_{12},M_{21},Z^{i-1},S^{i-1}).
\end{eqnarray}
Thus, we obtained
\begin{eqnarray}
    R_1&\leq& \frac{1}{n}\sum_{i=1}^n [H(Z_i|U_i) + I(X_{1,i};Y_i|U_i,X_{2,i},S_i,Z_i)] + C_{12} + \epsilon_n.
\end{eqnarray}
Next, we consider $R_2$;
\begin{eqnarray}
    nR_2 &=& H(M_2)\\
    &\stackrel{(a)}=& H(M_2|S^n,M_1)\\ 
    &\stackrel{(b)}=& H(M_{21},M_2|S^n,M_1)\\
    &=& H(M_{21}|S^n,M_1) + H(M_2|S^n,M_{21},M_1)\\
    &\stackrel{(c)}\leq& H(M_{21}|M_1) - I(M_{21};S^n|M_1) + I(M_2;Y^n|S^n,M_1,M_{21}) + n\epsilon_n \label{C21-1}\\ 
    &\stackrel{(d)}\leq& nC_{21} + \sum_{i=1}^n [I(M_2;Y_i|Y^{i-1},S^n,M_1,M_{21}) - I(S_i;M_{21}|S^{i-1},M_1)] + n\epsilon_n\\ 
    &\stackrel{(e)}=&  nC_{21} + \sum_{i=1}^n [I(M_2,X_{2,i};Y_i|Y^{i-1},M_1,M_{12},M_{21},S^n,X_{1,i},Z^{i-1})\notag\\ &&\tab - I(S_i;M_{21},S^{i-1},M_1,M_{12},Z^{i-1})] + n\epsilon_n\\ 
    &\stackrel{(f)}\leq& nC_{21} + \sum_{i=1}^n [I(X_{2,i};Y_i|M_{21},M_{12},S^{i},Z^{i-1},X_{1,i}) \notag\\ &&\tab - I(S_i;M_{21},S^{i-1},M_{12},Z^{i-1})] + n\epsilon_n\label{C21-2}\\
    &=& nC_{21} + \sum_{i=1}^n [I(X_{2,i};Y_i|U_i,S_i,V_i,X_{1,i}) - I(S_i;U_i)] + n\epsilon_n,
\end{eqnarray}
where (a) follows since $M_2$ is independent of $S^n$ and $M_1$, (b) follows since $M_{21}=f(S^n,M_2,M_1)$, (c) follows from Fano's inequality, (d) follows from the chain rule, (e) follows since $S_i$ is independent of $(S^{i-1},M_1)$ and since $(M_{12},Z^{i-1},X_{1,i}) = f(M_{21},M_1)$, and (f) follows from the same argument as in (\ref{markov-YX1X2S}) and since conditioning reduces entropy. Thus, we obtained
\begin{eqnarray}
    R_2 &\leq& C_{21} + \frac{1}{n}\sum_{i=1}^n [I(X_{2,i};Y_i|U_i,S_i,X_{1,i}) - I(S_i;U_i)] + \epsilon_n.
\end{eqnarray}
Now, consider
\begin{eqnarray}
    n(R_1 + R_2) &=& H(M_1,M_2,M_{12})\\
    &\stackrel{(a)}=& H(M_1,M_2|S^n,M_{12}) + H(M_{12})\\
    &\leq& H(M_1,M_2|M_{21},S^n,M_{12}) + H(M_{21}|S^n,M_{12}) + nC_{12}\\
    &\stackrel{(c)}\leq& nC_{12} + I(M_1,M_2,Z^n;Y^n|S^n,M_{12},M_{21}) + H(M_{21}|S^n,M_1) + n\epsilon_n\\ 
    &\leq& nC_{12} + I(M_1,M_2;Y^n|S^n,M_{12},M_{21},Z^n) + I(Z^n|S^n,M_{12},M_{21}) \notag\\ &&\tab + H(M_{21}|S^n,M_1) + n\epsilon_n\\ 
    &\stackrel{(d)}\leq& nC_{12} + I(X_1^n,X_2^n;Y^n|S^n,M_{12},M_{21},Z^n) + nC_{21}\notag\\ &&\tab +  \sum_{i=1}^n [H(Z_i|U_i) - I(S_i;U_i)]  + n\epsilon_n\\
    &\stackrel{(e)}=& nC_{12} + nC_{21} + \sum_{i=1}^n [I(X_1^n,X_2^n;Y_i|S^n,Y^{i-1},M_{12},M_{21},Z^n)\notag\\ &&\tab +  H(Z_i|U_i) - I(S_i;U_i)] + n\epsilon_n\\ 
    &\stackrel{(f)}=& nC_{12} + nC_{21} + \sum_{i=1}^n [I(X_{1,i},X_{2,i};Y_i|S_i,S^{i-1},M_{21},M_{12},Z^i)\notag\\ &&\tab +  H(Z_i|U_i) - I(S_i;U_i)] + n\epsilon_n\\ 
    &\leq& nC_{12} + nC_{21} + \sum_{i=1}^n [I(X_{1,i},X_{2,i};Y_i|S_i,U_i,Z_i)\notag\\ &&\tab +  H(Z_i|U_i)  - I(S_i;U_i)] + n\epsilon_n,
\end{eqnarray}
where (a) follows since $(M_1,M_2)$ is independent of $S^n$, (b) follows since $Z^n=f(M_1,M_{21})$, (c) follows from Fano's inequality and since $M_{21}$ is independent of $M_{12}$, (d) follows from the same arguments as given in (\ref{C21-1})-(\ref{C21-2}) and from the Markov chain $(M_1,M_2) - (X_1^n,X_2^n,M_{12},M_{21},Z^n,S^n) - Y^n$, (e) follows from the chain rule, and (f) follows from the same argument as given in (\ref{markov-YX1X2S}). Thus, we obtained
\begin{eqnarray}
    R_1 + R_2 &\leq& C_{12} + C_{21} + \frac{1}{n}\sum_{i=1}^n [I(X_{1,i},X_{2,i};Y_i|S_i,U_i,Z_i) +  H(Z_i|U_i) - I(S_i;U_i)] + \epsilon_n.
\end{eqnarray}
Additionally,
\begin{eqnarray}
    n(R_1 + R_2) &\leq& H(M_1,M_2)\\
    &\stackrel{(a)}=& H(M_1,M_2|S^n)\\
    &\stackrel{(b)}\leq& I(M_1,M_2;Y^n|S^n) + n\epsilon_n\\ 
    &\stackrel{(c)}\leq& I(X_1^n,X_2^n;Y^n|S^n) + n\epsilon_n\\ 
    &\stackrel{(d)}=& \sum_{i=1}^n I(X_1^n,X_2^n;Y_i|S^n,Y^{i-1}) + n\epsilon_n\\ 
    &\stackrel{(e)}\leq& \sum_{i=1}^n I(X_{1,i},X_{2,i};Y_i|S_i) + n\epsilon_n, 
\end{eqnarray}
where (a) follows since $(M_1,M_2)$ is independent of $S^n$, (b) follows from Fano's inequality, (c) follows from encoding relations (\ref{eq:en-relation-state1})-(\ref{eq:en-relation-state2}), (d) follows from the chain rule, and step (e) follows from the Markov Chain $Y_i - X_{1,i},X_{2,i},S_i - Y^{i-1}$ and since conditioning reduces entropy. Thus we obtained
\begin{eqnarray}
    R_1 + R_2 &\leq& \frac{1}{n}\sum_{i=1}^n I(X_{1,i},X_{2,i};Y_i|S_i) + \epsilon_n.
\end{eqnarray}
Finally,
\begin{eqnarray}
  nC_{21} &\geq& H(M_{21}) \\
   &\geq& H(M_{21}|M_{1}) \\
   &\geq& I(M_{21};S^n|M_{1}) \\
   &=& \sum_{i=1}^n I(S_i;M_{21}|S^{i-1},M_{1}) \\
   &\stackrel{(a)}=& \sum_{i=1}^n I(S_i;M_{21},S^{i-1},M_{1}) \\
   &\geq& \sum_{i=1}^n I(S_i;M_{21},S^{i-1},Z^{i-1},M_{12}) \\
   &=& \sum_{i=1}^n I(S_i;U_i),
\end{eqnarray}
where (a) follows since $S_i$ is independent of $(S^{i-1},M_1)$.
Finally, let $Q$ be an RV independent of $(X_1^n,X_2^n,Y^n)$ and uniformly distributed over the set $\{1,2,3,\dots,n\}$. We define the RV $U\triangleq(Q,U_Q)$ and obtain the region given in (\ref{eq:capacityS}).
\begin{figure}[!h]
    \begin{center}
        \begin{psfrags}
            \psfragscanon
            \psfrag{I}[][][1]{$M_1$}
            \psfrag{B}[][][1]{$M_{21}$}
            \psfrag{H}[][][1]{$Z^{i-1}$}
            \psfrag{A}[][][1]{$X_{1,i}$}
            \psfrag{F}[][][1]{$M_{12}$}
            \psfrag{D}[][][1]{$S^i$}
            \psfrag{C}[][][1]{$S_{i+1}^{n}$}
            \psfrag{E}[][][1]{$M_2$}
            \psfrag{G}[][][1]{$X_{2,i}$}
            \psfrag{J}[][][1]{$Z_i$}
            \includegraphics[scale = 0.7]{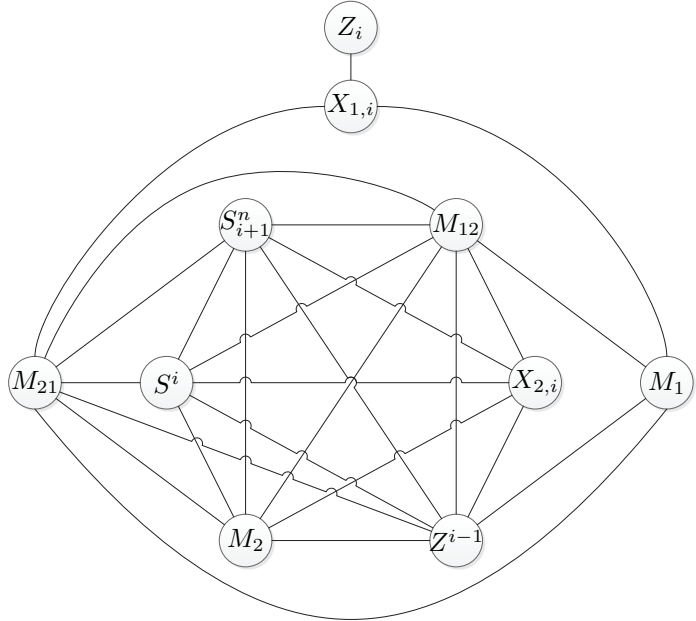}
            \caption{Proof of the Markov chains $Z_i - U_i - S_i$, $X_{1,i} - (U_i,Z_i) - S_i$, and $X_{2,i} - (M_{12},M_{21},Z^{i-1},S^{i-1}) - X_{1,i}$ using the undirected graphical technique \cite[Section II]{permuter2010two}. This graph corresponds to the joint distribution $P(s^n)P(m_1)P(m_2)P(m_{12}|m_1)P(m_{21}|m_2,s^n,m_{21})P(z^{i-1}|m_1,m_{21})P(x_{1,i}|m_{21},m_1)P(z_{i}|x_{1,i})P(x_{2,i}|m_{12},m_2,s^n,z^{i-1})$.} \label{fig:markov-state-A}
            \psfragscanoff
        \end{psfrags}
    \end{center}
\end{figure}

To complete the converse, we need to show the following Markov relations:
\begin{itemize}
\item $Z_i - U_i - S_i$, $X_{1,i} - (U_i,Z_i) - S_i$, and $X_{2,i} - (M_{12},M_{21},Z^{i-1},S^{i-1}) - X_{1,i}$ - These Markov relations can be proven by using the undirected graph method in Fig. \ref{fig:markov-state-A}. For the first Markov chain, see that it is impossible to get from node $Z_i$ to node $S_i$ without going through nodes $(S^{i-1},Z^{i-1},M_{12},M_{21})$. For the second Markov chain, it is impossible to get from node $X_{1,i}$ to node $S_i$ without going through nodes $(S^{i-1},Z^{i},M_{12},M_{21})$. Finally, for the third Markov chain, we can see that it is impossible to get from node $X_{1,i}$ to node $X_{2,i}$ without going through nodes $(S^{i},Z^{i-1},M_{12},M_{21})$.
\item $Y_i - (X_{1,i},X_{2,i}) - (Z_{1,i},U_i)$ - Follows from the fact that the channel output at any time $i$ is assumed to depend only on the channel inputs and state at time $i$.
\end{itemize}
This completes the converse part.
\hfill $\blacksquare$

\textit{Converse for the causal case:}
For the causal case we repeat the same converse as for the strictly causal case, except that in the final step we need to show the Markov chain $X_{2,i}-(U_i, Z_{i}, S_i)-X_{1,i}$, rather than $X_{2,i}-(U_i, S_i)-X_{1,i}$, as in the strictly causal case. If we change node $Z^{i-1}$ to $Z^i$ in Fig. \ref{fig:markov-state-A}, we can see that the Markov chain $X_{2,i} - (M_{12},M_{21},Z^{i},S^{i}) - X_{1,i}$ holds since we cannot get from node $X_{2,i}$ to node $X_{1,i}$ without going through nodes $(M_{12},M_{21},Z^{i},S^{i})$.
\hfill $\blacksquare$

\subsection{Achievability}
In order to prove the achievability, we will consider a similar setting and then, by doing minor modifications, we will prove our setting. We first prove the achievability for the strictly causal case.

\textit{Achievability for the strictly causal case:}
Let us look at a similar model depicted in Fig. \ref{fig:common+cribbing+decoder}.

\begin{figure}[!h]
\begin{center}
\begin{psfrags}
    \psfragscanon
    \psfrag{A}[][][1]{Encoder 1}
    \psfrag{B}[][][1]{Encoder 2}
    \psfrag{D}[][][1]{$S^n$}
    \psfrag{E}[][][1]{$P_{Y|X_1,X_2,S}$}
    \psfrag{F}[][][1]{Decoder}
    \psfrag{G}[][][1]{$m_0,m_2$}
    \psfrag{H}[][][1]{$m_0,m_1$}
    \psfrag{I}[][][1]{$X_{1,i}(m_0,m_1,m_{21})$}
    \psfrag{J}[][][1]{$X_{2,i}(m_0,m_2,Z^{i-1},S^n)$}
    \psfrag{K}[][][1]{$Y_i$}
    \psfrag{L}[][][1]{$(\hat{m}_0,\hat{m_1},\hat{m_2})$}
    \psfrag{M}[][][1]{$C_{21}$}
    \psfrag{O}[][][1]{$Z_i=f(X_{1,i})$}
\includegraphics[width=16cm]{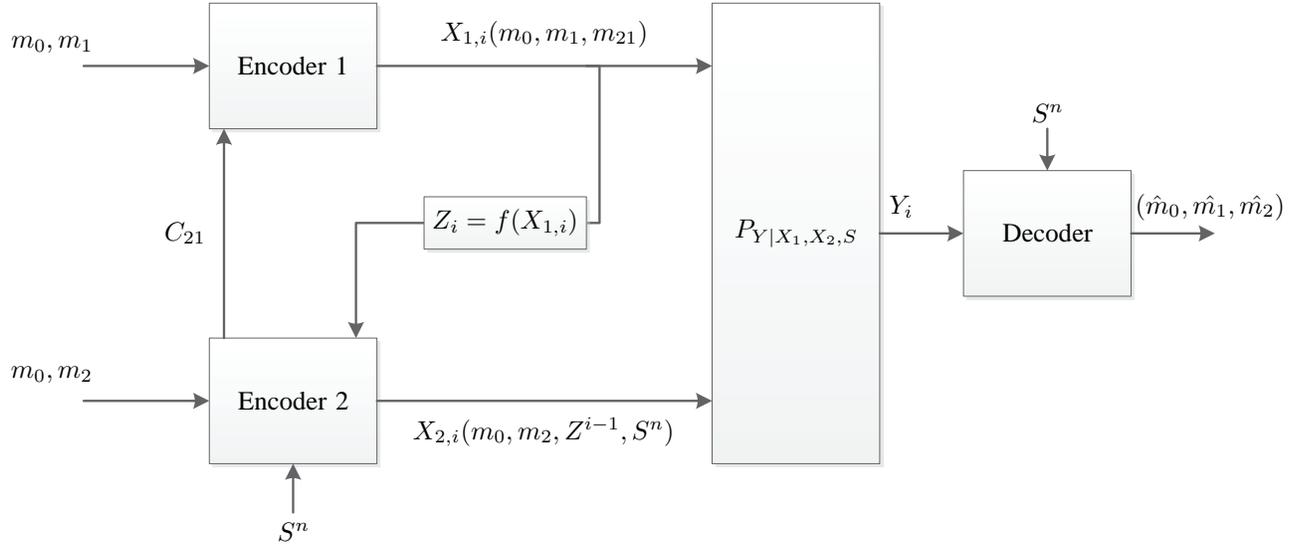}
\caption{MAC with a common message and state known at a partially cribbing Encoder.} \label{fig:common+cribbing+decoder}
\psfragscanoff
\end{psfrags}
\end{center}\vspace{-10mm}
\end{figure}

First, we will solve the achievability for this model.
Fix a joint distribution $P(s)P(u|s)P(z,x_1|u)P(x_2|s,u)P(y|x_1,x_2,s)$ where $P(s)$ and $P(y|x_1,x_2,s)$ are given by the channel. In the following achievability scheme we use block Markov coding, rate splitting, and double binning.

\emph{\underline{Coding Scheme:}} We consider $B$ blocks, each consisting of $n$ symbols; thus we transmit $nB$ symbols. We transmit $B-1$ messages $M_1$ in $B$ blocks of information. Here, $M_1 \in \{1,\dots,2^{nR_1}\}$; thus asymptotically, for a large enough $n$, our transmission rate would be $\frac{nR_1(B-1)}{nB}\stackrel{n\rightarrow\infty}\longrightarrow R_1$. At each block we split messages $M_1$ and $M_2$ into $(M^{\prime}_1,M^{\prime\prime}_1)$ and $(M^{\prime}_2,M^{\prime\prime}_2)$ at rates $(R^{\prime}_1,R^{\prime\prime}_1)$ and $(R^{\prime}_2,R^{\prime\prime}_2)$, respectively. We note that $R^{\prime}_1+R^{\prime\prime}_1=R_1$ and $R^{\prime}_2+R^{\prime\prime}_2=R_2$.

\emph{\underline{Code Design:}} The following binning process is depicted in Fig. \ref{fig:binning}. Generate $2^{n(R_0 + R^{\prime}_1 + C_{21})}$ codewords $u^n$ i.i.d. using $P(u^n)=\Pi_{i=1}^nP(u_i)$. Bin all $u^n$s into $2^{n(R_0 + R^{\prime}_1)}$ super-bins. In each super-bin, bin all $u^n$s into $2^{nR_2^\prime}$ bins. Thus we have $2^{n(R_0 + R^{\prime}_1)}$ super-bins, each consisting of  $2^{nR_2^\prime}$ bins, where in each bin we have $2^{n(C_{21}-R_2^\prime)}$ $u^n$ codewords. For each $u^n$, generate $2^{nR_1^\prime}$ codewords $z^n$ i.i.d. using $P(z^n|u^n)=\Pi_{i=1}^nP(z_i|u_i)$. For each pair $(u^n,z^n)$, generate $2^{nR_1^{\prime\prime}}$ codewords $x_1^n$ i.i.d. using $P(x_1^n|u^n,z^n)=\Pi_{i=1}^nP(x_{1,i}|u_i,z_i)$. Additionally, for each pair $(u^n,s^n)$, generate $2^{nR_2^{\prime\prime}}$ codewords $x_2^n$ i.i.d. using $P(x_2^n|u^n,s^n)=\Pi_{i=1}^nP(x_{2,i}|u_i,s_i)$.
\begin{figure}[!h]
\begin{center}
\begin{psfrags}
    \psfragscanon
    \psfrag{A}[][][1]{a codeword $u^n$}
    \psfrag{B}[][][1]{a superbin (contains bins)}
    \psfrag{C}[][][1]{a bin (contains codewords)}
    \includegraphics[width=16cm]{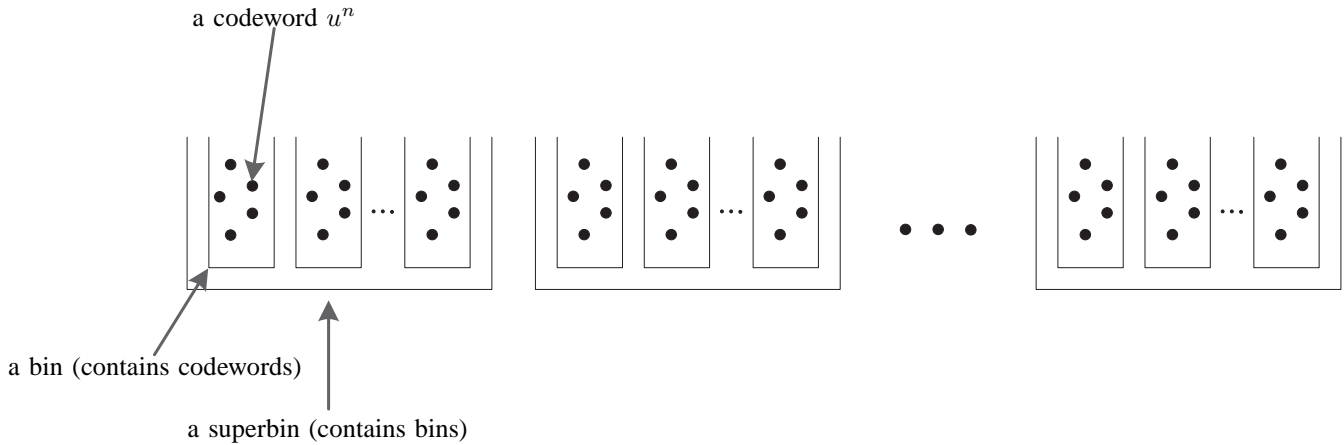}
\caption{The binning process as explained in the code design. There are $2^{n(R_0 + R^{\prime}_1)}$ super-bins and $2^{nR_2^\prime}$ bins in each super-bin. The number of codewords in each bin must be greater than $I(U;S)$ in order to find $u^n$ such that $(u^n,s^n)\in T_\epsilon^{(n)}(U,S)$.} \label{fig:binning}
\psfragscanoff
\end{psfrags}
\end{center}
\end{figure}

\emph{\underline{Encoding:}} We denote the realizations of the sequences $(M_0,M_1^\prime,M_1^{\prime\prime},M_2^\prime,M_2^{\prime\prime})$ at block $b$ as $(m_{0,b},m_{1,b}^\prime,m_{1,b}^{\prime\prime},m_{2,b}^\prime,m_{2,b}^{\prime\prime})$.
Since we use block Markov coding, we set $m_{1,B}^\prime=1$. In block $b \in \{1,\dots,B\}$, Encoder 2 looks in super-bin $(m_{0,b},m_{1,b}^\prime)$ and bin $m_{2,b}^\prime$ for $u^n$ such that $(u^n,s^n)\in T_\epsilon^{(n)}(U,S)$ and sends its index $l$ inside the super-bin over the rate-limited cooperation link to Encoder 1, where $l\in\{1,\dots,2^{nC_{21}}\}$. If such a codeword $u^n$ does not exist, namely, among the codewords in the bin none is jointly typical with $s^n$, choose an arbitrary $u^n$ from the bin $m_{2,b}^\prime$ (in such a case the decoder will declare an error). Encoder 1 looks in super-bin $(m_{0,b},m_{1,b}^\prime)$ for the bin that $u^n(l)$ lies in. That bin's index is $m_{2,b}^\prime$. Then, Encoder 1 encodes message $m_{1,b}^\prime$ conditioned on $(m_{0,b},m_{1,b-1},m_{2,b}^\prime)$ using $z^n(m_{1,b}^{\prime},u^n)$ and encodes message $m_{1,b}^{\prime\prime}$ conditioned on $(m_{0,b},m_{1,b-1},m_{2,b}^\prime,m_{1,b}^\prime)$ using $x_1^n(m_{1,b}^{\prime\prime},u^n,z^n)$. Encoder 2 encodes message $m_{2,b}^{\prime\prime}$ conditioned on $(m_{0,b},m_{1,b-1},m_{2,b}^\prime)$ and $s^n$ using $x_2^n(m_{2,b},u^n,s^n)$. Send $x_1^n(m_{1,b}^{\prime\prime},u^n,z^n)$ and $x_2^n(m_{2,b}^{\prime\prime},u^n,s^n)$ over the channel.

\emph{\underline{Decoding at Encoder 2:}} At the end of block $b$, Encoder 2 tries to decode message $m^\prime_{1,b}$. Given $(m_{0,b},m^\prime_{2,b})$ and assuming that message $m^\prime_{1,b-1}$ was decoded correctly at the end of block $b-1$, Encoder 2 looks for $\hat{m}^\prime_{1,b}$ s.t.
\begin{equation}
    (u^n(m_{0,b},m^\prime_{1,b-1},m^\prime_{2,b}),z^n(\hat{m}^\prime_{1,b},u^n))\in T_\epsilon^{(n)}(U,Z).
\end{equation}
If no such $\hat{m}^\prime_{1,b}$, or more than one such $\hat{m}^\prime_{1,b}$, was found, an error is declared at block $b$ and therefore in the whole super-block $nB$.

\emph{\underline{Decoding at the receiver:}} At the end of block $B$, the decoding is done backwards. At block $b$, assuming that $(m_{0,b+1},m_{1,b},m_{2,b+1}^\prime)$ was decoded correctly in block $b+1$, the decoder looks for the set $(\hat{m}_{0,b},\hat{m}_{1,b-1}^{\prime},\hat{m}_{1,b}^{\prime\prime},\hat{m}_{2,b}^{\prime},\hat{m}_{2,b}^{\prime\prime})$ s.t.
\begin{align*}
    (u^n(\hat{m}_{0,b},\hat{m}_{1,b-1}^\prime,\hat{m}_{2,b}^\prime,s^n),z^n(\hat{m}_{1,b}^\prime,u^n),x_1^n(m_{1,b}^{\prime\prime},u^n,z^n),x_2^n(\hat{m}_{2,b}^{\prime\prime},u^n,s^n),s^n,y^n)\in T_\epsilon^{(n)}(U,Z,X_1,X_2,S,Y).
\end{align*}
If no such tuple, or more than one such tuple, was found, an error is declared in block $b$ and therefore at the whole super-block $nB$.

\emph{\underline{Error Analysis:}} The probability that $z^n(1,u^n)=z^n(i,u^n)$, where $i>1$ and where $(u^n,z^n(1,u^n))\in T_\epsilon^{(n)}(U,Z)$ is bounded by $2^{-n(H(Z|U)-\delta(\epsilon))}$, where $\delta(\epsilon)$ goes to zero as $\epsilon$ goes to zero. Hence, if
\begin{equation}
    R_1^\prime<H(Z|U),
\end{equation}
then the probability that an incorrect message $m_{1,b}^\prime$ was decoded goes to zero for a large enough $n$.
In order to find in super-bin $(\hat{m}_{0,b},\hat{m}_{1,b-1}^\prime)$ and in bin $m_{2,b}^\prime$ a codeword $u^n$ that is jointly typical with $s^n$, we need to have more than $I(U;S)$ codewords in each bin; thus if
\begin{eqnarray}
  C_{21}-R_2^\prime &\geq& I(U;S), \\
  R_2^\prime &\leq& C_{21}-I(U;S),
\end{eqnarray}
then the probability of finding a codeword $u^n$ such that $(u^n,s^n)\in T_\epsilon^{(n)}(U,S)$ goes to 1 for a large enough $n$.
We define the following event at block b:
\begin{equation}
    E_{i,j,k,b} \triangleq (u^n(i,s^n),z^n(\hat{m}_{1,b}^\prime,u^n),x_1^n(j,u^n,z^n),x_2^n(k,s^n),s^n,y^n)\in T_\epsilon^{(n)}(U,Z,X_1,X_2,S,Y).
\end{equation}
We can bound the probability of error as follows:
\begin{eqnarray}
    P_{e,b}^{(n)} &\leq& \Pr(E^c_{1,1,1,b}) + \sum_{i=1,j=1,k>1}\Pr(E_{1,1,k,b}) +  \sum_{i=1,j>1,k=1}\Pr(E_{1,j,1,b})\notag\\ &&+ \sum_{i=1,j>1,k>1}\Pr(E_{1,j,k,b}) + \sum_{i>1,j>1,k>1}\Pr(E_{i,j,k,b}). \label{eq:prob_of_error1}
\end{eqnarray}
We now show that each term in (\ref{eq:prob_of_error1}) goes to zero for a large enough $n$.
\begin{itemize}
  \item Upper-bounding $\Pr(E^c_{1,1,1,b})$: Since we assume that Encoders 1 and 2 encode the correct message-tuple $(m_{0,b},m_{1,b-1}^\prime,m_{1,b}^{\prime\prime},m_{2,b}^{\prime},m_{2,b}^{\prime\prime})$ at block $b$ and that the decoder decoded the right $(m_{0,b+1},m_{1,b}^\prime,m_{1,b+1}^{\prime\prime},m_{2,b+1}^{\prime},m_{2,b+1}^{\prime\prime})$ at block $b+1$, by the LLN, $\Pr(E^c_{1,1,1,b}) \rightarrow 0$.
  \item Upper-bounding $\sum_{i=1,j=1,k>1}\Pr(E_{1,1,k,b})$: Assuming that $m_{1,b}^\prime$ was decoded correctly at block $b+1$, the probability for this event is bounded by
      \begin{eqnarray}
            \sum_{i=1,j=1,k>1}\Pr(E_{1,1,k,b}) &\leq& 2^{nR_2^{\prime\prime}}2^{-n(I(X_2;Y|S,U,Z,X_1)-\delta(\epsilon)}\\
            &=& 2^{nR_2^{\prime\prime}}2^{-n(I(X_2;Y|S,U,X_1)-\delta(\epsilon)}.
      \end{eqnarray}
  \item Upper-bounding $\sum_{i=1,j>1,k=1}\Pr(E_{1,j,1,b})$: Assuming that $m_{1,b}^\prime$ was decoded correctly at block $b+1$, the probability for this event is bounded by
      \begin{equation}
            \sum_{i=1,j>1,k=1}\Pr(E_{1,j,1,b}) \leq 2^{n(R_1^{\prime\prime})}2^{-n(I(X_1;Y|S,U,Z,X_2)-\delta(\epsilon)}.
      \end{equation}
  \item Upper-bounding $\sum_{i=1,j>1,k>1}\Pr(E_{1,j,k,b})$: Assuming that $m_{1,b}^\prime$ was decoded correctly at block $b+1$, the probability for this event is bounded by
      \begin{eqnarray}
            \sum_{i=1,j>1,k>1}\Pr(E_{1,j,k,b}) &\leq& 2^{n(R_1^{\prime\prime}+R_2^{\prime\prime})}2^{-n(I(X_1,X_2;Y|S,U,Z)-\delta(\epsilon)}.
      \end{eqnarray}
  \item Upper-bounding $\sum_{i>1,j>1,k>1}\Pr(E_{i,j,k,b})$: Assuming that $m_{1,b}^\prime$ was decoded correctly at block $b+1$, the probability for this event is bounded by
      \begin{eqnarray}
            \sum_{i>1,j>1,k>1}\Pr(E_{1,j,k,b}) &\leq& 2^{n(R_0+R_1^\prime+R_1^{\prime\prime}+R_2^{\prime}+R_2^{\prime\prime})}2^{-n(I(U,V,Z,X_1,X_2;Y|S)-\delta(\epsilon)}\\
            &\leq& 2^{n(R_0+R_1^\prime+R_1^{\prime\prime}+R_2^{\prime}+R_2^{\prime\prime})}2^{-n(I(X_1,X_2;Y|S)-\delta(\epsilon)}.
      \end{eqnarray}
\end{itemize}
To summarize, we note that $R_1^\prime=R_1-R_1^{\prime\prime}$ and $R_2^\prime=R_2-R_2^{\prime\prime}$ and thus we obtained that if $(R_1^{\prime\prime},R_2^{\prime\prime},R_1,R_2)$ satisfy
\begin{eqnarray}
  R_1-R_1^{\prime\prime} &\leq& H(Z|U), \\
  R_2-R_2^{\prime\prime} &\leq& C_{21}-I(U;S), \\
  R_2^{\prime\prime} &\leq& I(X_2;Y|S,U,X_1), \\
  R_1^{\prime\prime} &\leq& I(X_1;Y|S,U,Z,X_2), \\
  R_1^{\prime\prime}+R_2^{\prime\prime} &\leq& I(X_1,X_2;Y|S,U,Z), \\
  R_0+R_1^{\prime\prime}+R_2^{\prime\prime} &\leq& I(X_1,X_2;Y|S),
\end{eqnarray}
then there exists a code with a probability of error that goes to zero as the block length goes to infinity.
Using the Fourier-Motzkin elimination and by setting $R_1 = \tilde{R_1}, R_0 = \tilde{R_0}$, we obtain the following region
\begin{eqnarray}
    \tilde{R_1} &\leq& H(Z |V,U)+I(X_1;Y|S,U,X_2,Z),\notag\\
    R_2 &\leq& I(X_2; Y |X_1, S,U) + C_{21} - I(U;S),\notag\\
    \tilde{R_1} + R_2 &\leq& I(X_1, X_2 ; Y|U, Z, S)+H(Z|U)  + C_{21} - I(U;S),\notag \\
    \tilde{R_0} + \tilde{R_1} + R_2 &\leq& I(X_1, X_2 ; Y|S).\label{eq:capacityS_old}
\end{eqnarray}
Now we can easily see that if we set
\begin{eqnarray}
    \tilde{R_0} = C_{12},\\
    \tilde{R_1} = R_1-C_{12},
\end{eqnarray}
then the inequalities can be rewritten as
\begin{eqnarray}
    R_1 - C_{12} &\leq& H(Z |U)+I(X_1;Y|S,U,X_2,Z),\notag\\
    R_2 &\leq& I(X_2; Y |X_1, V, S,U) + C_{21} - I(U;S),\notag\\
    R_1 - C_{12} + R_2 &\leq& I(X_1, X_2 ; Y|U, Z, S)+H(Z|U)  + C_{21} - I(U;S),\notag\\
    C_{12} + (R_1 - C_{12}) + R_2 &\leq& I(X_1, X_2 ; Y|S),
\end{eqnarray}
and thus we obtain the region in (\ref{eq:capacityS}).
\hfill $\blacksquare$

\textit{Achievability for the causal case:}
The achievability part follows similarly to that of the strictly causal case, only now the generation of $X_2^n$ is done i.i.d. according to the conditional distribution of $p(x_2|u,s,z)$ induced by (\ref{eq:state-joint-B}).

\hfill $\blacksquare$

\section{Proof of Theorem \ref{theorem:add1}} \label{appendix:crib+coop+state}
\subsection{Converse}
\textit{Converse for the strictly causal case:}
Given an achievable rate-pair $(R_1,R_2)$, we need to show that there exists a joint distribution of the form $P(w)P(v|w)p(a|w)P(s|a)P(x_1|v,w)P(u,x_2|s,v,a,w)P(y|x_1,x_2,s)$ such that the inequalities in (\ref{eq:capacity8}) are satisfied. Since $(R_1,R_2)$ is an achievable rate-pair, there exists a $(2^{nR_1},2^{nR_2},2^{nC_{12}},n)$ code with an arbitrarily small error probability $P^{(n)}_e$. By Fano's inequality,
\begin{eqnarray}
H(M_1,M_2|Y^n)\leq n(R_1 + R_2)P^{(n)}_e + H(P^{(n)}_e).
\end{eqnarray}
We set
\begin{equation}
(R_1 + R_2)P^{(n)}_e + \frac{1}{n}H(P^{(n)}_e)\triangleq \epsilon_n,
\end{equation}
where $\epsilon_n \rightarrow 0$ as $P^{(n)}_e\rightarrow 0$. Hence,
\begin{eqnarray}
H(M_1|Y^n,M_2)\leq H(M_1,M_2|Y^n)\leq n\epsilon_n,\\
H(M_2|Y^n,M_1)\leq H(M_1,M_2|Y^n)\leq n\epsilon_n.
\end{eqnarray}
For $R_1$ we have the following:
\begin{eqnarray}
    nR_1 &=& H(M_1)\\
        &=& H(M_1,M_{12})\\
        &\stackrel{(a)}=&   H(M_1|M_2,M_{12}) + H(M_{12})\\
        &\leq&                 nC_{12} + I(M_1;Y^n|M_2,M_{12}) + H(M_1|Y^n,M_2,M_{12})\\
        &\stackrel{(b)}\leq&nC_{12} + I(M_1;Y^n|M_2,M_{12}) + n\epsilon_n\\
        &\stackrel{(c)}=&   nC_{12} + I(X^n_1;Y^n|M_2,M_{12}) + n\epsilon_n\\
        &\stackrel{(d)}=&   nC_{12} + \sum_{i=1}^n I(X_{1,i};Y^n|M_2,X^{i-1}_1,M_{12}) + n\epsilon_n\\
        &\leq&              nC_{12} + \sum_{i=1}^n H(X_{1,i}|M_2,X^{i-1}_1,M_{12}) + n\epsilon_n\\
        &\stackrel{(e)}\leq&nC_{12} + \sum_{i=1}^n H(X_{1,i}|V_i,W_i) + n\epsilon_n,
\end{eqnarray}
where (a) follows from the fact that the messages $M_1$ and $M_2$ are independent, (b) follows from Fano's inequality, (c) follows from the encoding relation in (\ref{eq:f1}), (d) follows from the chain rule, and step (e) follows since conditioning reduces entropy and by setting the RVs
\begin{eqnarray}
    V_i&\triangleq& X_1^{i-1},\\
    W_i&\triangleq& M_{12}.
\end{eqnarray}
Thus, we obtained
\begin{eqnarray}
    R_1&\leq& C_{12} + \frac{1}{n}\sum_{i=1}^n H(X_{1,i}|V_i,W_i) + \epsilon_n.
\end{eqnarray}
Additionally,
\begin{eqnarray}
    nR_1 &=& H(M_1)\\
    &=& H(M_1|M_2,M_{12}) + H(M_{12}) \\ 
    &\stackrel{(a)}\leq& nC_{12} + I(M_1;Y^n|M_2,M_{12}) + n\epsilon_n\\ 
    &\stackrel{(b)}=& nC_{12} + \sum_{i=1}^n I(M_1;Y_i|Y^{i-1},M_2,M_{12}) + n\epsilon_n\\ 
    &\leq& nC_{12} + \sum_{i=1}^n I(Y^{i-1},M_1,M_2;Y_i|M_{12}) + n\epsilon_n\\
    &=& nC_{12} + \sum_{i=1}^n [I(Y^{i-1},M_1,M_2,S_{i+1}^n;Y_i|M_{12}) \notag\\ &&\tab - I(S_{i+1}^n;Y_i|M_1,M_2,Y^{i-1},M_{12})] + n\epsilon_n\\ 
    &\stackrel{(c)}=& nC_{12} + \sum_{i=1}^n [I(Y^{i-1},M_1,M_2,S_{i+1}^n,X_1^{i-1},X_{1,i};Y_i|M_{12}) \notag\\ &&\tab - I(S_{i};Y^{i-1}|M_1,M_2,S_{i+1}^n,M_{12})] + n\epsilon_n \label{eq:csiszar1}\\
    &\stackrel{(d)}=& nC_{12} + \sum_{i=1}^n [I(Y^{i-1},M_2,S_{i+1}^n,X_1^{i-1},X_{1,i};Y_i|A_i,M_{12}) \notag\\ &&\tab - I(S_{i};Y^{i-1}|M_1,M_2,A_i,S_{i+1}^n,M_{12})] + n\epsilon_n \label{eq:markov-state3}\\
    &\stackrel{(e)}\leq& nC_{12} + \sum_{i=1}^n [I(Y^{i-1},M_2,S_{i+1}^n,X_1^{i-1},X_{1,i};Y_i|A_i,M_{12}) \notag\\ &&\tab - I(S_{i};Y^{i-1},M_2,S_{i+1}^n|M_1,A_i,M_{12})] + n\epsilon_n \\ 
    &=& nC_{12} + \sum_{i=1}^n [I(Y^{i-1},M_2,S_{i+1}^n,X_1^{i-1},X_{1,i};Y_i|A_i,M_{12}) \notag\\ &&\tab - I(S_{i};Y^{i-1},M_2,S_{i+1}^n|M_1,A_i,X_1^{i-1},M_{12})] + n\epsilon_n \\
    &\stackrel{(f)}=& nC_{12} + \sum_{i=1}^n [I(Y^{i-1},M_2,S_{i+1}^n,X_1^{i-1},X_{1,i};Y_i|A_i,M_{12}) \notag\\ &&\tab - I(S_{i};Y^{i-1},M_2,S_{i+1}^n|A_i,X_1^{i-1},M_{12})] + n\epsilon_n\\
    &\stackrel{(g)}=& nC_{12} + \sum_{i=1}^n [I(V_i,U_i,X_{1,i};Y_i|A_i,W_i) - I(S_{i};U_i|V_i,A_i,W_i)] + n\epsilon_n,\label{eq:csiszar2} 
\end{eqnarray}
where (a) follows from Fano's inequality, (b) follows from the chain rule, (c) follows since $X_1^i=f(M_1)$ and by using the Csiszar Sum Equality, (d) follows since $A_i = f(M_{12},M_2)$ and from the Markov Chain $M_1 - (M_{12},X_{1,i},X_1^{i-1},Y^{i-1},M_2,S_{i+1}^n,A_i,M_{12}) - Y_i$, (e) follows since $S_i$ is independent of $(M_2,S_{i+1}^n)$ given $(M_1,A_i)$, (f) follows from the Markov Chain $M_1 - (M_{12},X_1^{i-1},A_i) - (Y^{i-1},M_2,S_{i+1}^n)$, and (g) follows by setting the RVs $W$,$V$ and
\begin{eqnarray}
    U_i\triangleq(Y^{i-1},M_2,S_{i+1}^n).
\end{eqnarray}
Thus, we obtained
\begin{eqnarray}
    R_1&\leq& C_{12} + \frac{1}{n}\sum_{i=1}^n [I(V_i,U_i,X_{1,i};Y_i|A_i,W_i) - I(S_{i};U_i|V_i,A_i,W_i)] + \epsilon_n.
\end{eqnarray}
Next, we consider $R_2$
\begin{eqnarray}
    nR_2 &=& H(M_2)\\
    &=& H(M_2|M_1)\\ 
    &\stackrel{(a)}\leq& I(M_2;Y^n|M_1) + n\epsilon_n\\ 
    &=& \sum_{i=1}^n I(M_2;Y_i|Y^{i-1},M_1) + n\epsilon_n\\ 
    &\leq& \sum_{i=1}^n I(Y^{i-1},M_2;Y_i|M_1) + n\epsilon_n\\
    &\stackrel{(b)}=& \sum_{i=1}^n [I(Y^{i-1},M_2,S_{i+1}^n;Y_i|M_1) - I(S_{i+1}^n;Y_i|M_1,M_2,Y^{i-1})] + n\epsilon_n\\ 
    &\stackrel{(c)}=& \sum_{i=1}^n [I(Y^{i-1},M_2,S_{i+1}^n;Y_i|M_1,M_{12},X_{1,i},X_1^{i-1}) \notag\\ &&\tab - I(S_{i};Y^{i-1},M_2,S_{i+1}^n|A_i,M_{12},X_1^{i-1})] + n\epsilon_n \\
    &\stackrel{(d)}=& \sum_{i=1}^n [I(Y^{i-1},M_2,S_{i+1}^n,A_i;Y_i|M_1,M_{12},X_{1,i},X_1^{i-1}) \notag\\ &&\tab - I(S_{i};Y^{i-1},M_2,S_{i+1}^n|A_i,M_{12},X_1^{i-1})] + n\epsilon_n \\
    &\stackrel{(e)}\leq& \sum_{i=1}^n [I(Y^{i-1},M_2,S_{i+1}^n,A_i;Y_i|M_{12},X_{1,i},X_1^{i-1}) \notag\\ &&\tab - I(S_{i};Y^{i-1},M_2,S_{i+1}^n|A_i,M_{12},X_1^{i-1})] + n\epsilon_n \\ 
    &\stackrel{(f)}=& \sum_{i=1}^n [I(U_i,A_i;Y_i|W_i,X_{1,i},V_i) - I(S_{i};U_i|W_i,V_i,A_i)] + n\epsilon_n, 
\end{eqnarray}
where (a) follows from Fano's inequality, (b) follows from the chain rule, (c) follows since $(M_{12},X_1^i)=f(M_1)$ and from the same arguments as given in (\ref{eq:csiszar1}) - (\ref{eq:csiszar2}), (d) follows since $A_i = f(M_{12},M_2)$, (e) follows from the same arguments as given in (\ref{eq:markov-state3}), and (f) follows by setting the RVs $U,V$ and $W$. Thus, we obtained
\begin{eqnarray}
    R_2 &\leq& \frac{1}{n}\sum_{i=1}^n [I(U_i,A_i;Y_i|W_i,X_{1,i},V_i) - I(S_{i};U_i|W_i,V_i,A_i)] +\epsilon_n.
\end{eqnarray}
Now, consider
\begin{eqnarray}
    n(R_1 + R_2) &=& H(M_1,M_2)\\
    &=& H(M_1,M_2|M_{12})+H(M_{12})\\
    &\stackrel{(a)}\leq& nC_{12} + I(M_1,M_2;Y^n|M_{12}) + n\epsilon_n\\ 
    &\stackrel{(b)}=& nC_{12} + \sum_{i=1}^n I(M_1,M_2;Y_i|Y^{i-1},M_{12}) + n\epsilon_n\\ 
    &\stackrel{(c)}\leq& nC_{12} + \sum_{i=1}^n [I(M_1,Y^{i-1},M_2,S_{i+1}^n;Y_i|M_{12}) \notag\\ &&\tab - I(Y^{i-1},M_2,S_{i+1}^n;S_i|M_{12},A_i,X_1^{i-1})] + n\epsilon_n\\ 
    &\stackrel{(d)}=& nC_{12} + \sum_{i=1}^n [I(M_1,X_{1,i},X_1^{i-1},Y^{i-1},M_2,S_{i+1}^n;Y_i|M_{12}) \notag\\ &&\tab - I(Y^{i-1},M_2,S_{i+1}^n;S_i|M_{12},A_i,X_1^{i-1})] + n\epsilon_n\\ 
    &\stackrel{(e)}=& nC_{12} + \sum_{i=1}^n [I(X_{1,i},X_1^{i-1},Y^{i-1},M_2,S_{i+1}^n,A_i;Y_i|M_{12}) \notag\\ &&\tab - I(Y^{i-1},M_2,S_{i+1}^n;S_i|M_{12},A_i,X_1^{i-1})] + n\epsilon_n\\ 
    &=& nC_{12} + \sum_{i=1}^n [I(U_i,V_i,X_{1,i},A_i;Y_i|W_i) - I(U_i;S_i|V_i,A_i,W_i)] + n\epsilon_n, 
\end{eqnarray}
where (a) follows from Fano's inequality, (b) follows from the chain rule, (c) follows from the same arguments as given in (\ref{eq:csiszar1})-(\ref{eq:csiszar2}), (d) follows since $X_1^i=f(M_1)$ and $A_i = f(M_{12},M_2)$, and (e) follows from the same arguments as given in (\ref{eq:markov-state3}). Thus we obtained
\begin{eqnarray}
    R_1 + R_2 &\leq& C_{12} + \frac{1}{n}\sum_{i=1}^n [I(U_i,V_i,X_{1,i},A_i;Y_i|W_i) - I(U_i;S_i|V_i,A_i|W_i)] + \epsilon_n.
\end{eqnarray}
Again,
\begin{eqnarray}
    n(R_1 + R_2) &=& H(M_1,M_2)\\
    &\stackrel{(a)}\leq& I(M_1,M_2;Y^n) + n\epsilon_n\\ 
    &\stackrel{(b)}=& \sum_{i=1}^n I(M_1,M_2;Y_i|Y^{i-1}) + n\epsilon_n\\ 
    &\stackrel{(c)}\leq& \sum_{i=1}^n [I(M_1,Y^{i-1},M_2,S_{i+1}^n;Y_i) \notag\\ &&\tab - I(Y^{i-1},M_2,S_{i+1}^n;S_i|M_{12},A_i,X_1^{i-1})] + n\epsilon_n\\ 
    &\stackrel{(d)}=& \sum_{i=1}^n [I(M_1,M_{12},X_{1,i},X_1^{i-1},Y^{i-1},M_2,S_{i+1}^n;Y_i) \notag\\ &&\tab - I(Y^{i-1},M_2,S_{i+1}^n;S_i|M_{12},A_i,X_1^{i-1})] + n\epsilon_n\\ 
    &\stackrel{(e)}=& \sum_{i=1}^n [I(M_{12},X_{1,i},X_1^{i-1},Y^{i-1},M_2,S_{i+1}^n,A_i;Y_i) \notag\\ &&\tab - I(Y^{i-1},M_2,S_{i+1}^n;S_i|M_{12},A_i,X_1^{i-1})] + n\epsilon_n\\ 
    &\leq& \sum_{i=1}^n [I(W_i,U_i,V_i,X_{1,i},A_i;Y_i) - I(U_i;S_i|W_i,V_i,A_i)] + n\epsilon_n, 
\end{eqnarray}
where (a) follows from Fano's inequality, (b) follows from the chain rule, (c) follows from the same arguments as given in (\ref{eq:csiszar1})-(\ref{eq:csiszar2}), (d) follows since $X_1^i=f(M_1)$ and $A_i = f(M_{12},M_2)$, and (e) follows from the same arguments as given in (\ref{eq:markov-state3}). Thus, we obtained
\begin{eqnarray}
    R_1 + R_2 &\leq& \frac{1}{n}\sum_{i=1}^n [I(W_i,U_i,V_i,X_{1,i},A_i;Y_i) - I(U_i;S_i|W_i,V_i,A_i)] + \epsilon_n.
\end{eqnarray}
Finally, we need to prove the following Markov chains:
\begin{itemize}
\item $A_i - W_i - V_i$ -
\begin{eqnarray}
    p(a_i|m_{12},x_1^{i-1}) &=& \sum_{m_2\in \mathcal{M}_2} p(m_2|m_{12},x_1^{i-1})p(a_i|m_{12},m_2,x_1^{i-1})\\
    &\stackrel{(a)}=&\sum_{m_2\in \mathcal{M}_2} p(m_2|m_{12})p(a_i|m_{12},m_2)\\
    &=&p(a_i|m_{12}),
\end{eqnarray}
where (a) follows since $m_2$ is independent of $m_1$ and since $a_i = f(m_2,m_{12})$.
\item $S_i - A_i - (W_i,V_i)$ - Follows from the fact that the channel state at any time $i$ is assumed to depend only on the action at time $i$.
\item $X_{1,i} - (V_i,W_i) - (A_i,S_i)$ -
\begin{eqnarray}
    p(x_{1,i}|m_{12},x_1^{i-1},a_i,s_i) &=& \sum_{m_1\in \mathcal{M}_1} p(m_1|m_{12},x_1^{i-1},a_i,s_i)p(x_{1,i}|m_{12},m_1,x_1^{i-1},a_i,s_i)\\
    &\stackrel{(a)}=&\sum_{m_1\in \mathcal{M}_1} p(m_1|m_{12},x_1^{i-1})p(x_{1,i}|m_1,m_{12},x_1^{i-1})\\
    &=&p(x_{1,i}|m_{12},x_1^{i-1}),
\end{eqnarray}
where (a) follows since $m_1$ is independent of $(a_i,s_i)$ given $(m_{12},x_1^{i-1})$ and since $x_{1,i} = f(m_1)$.
\item $(U_i,X_{2,i}) - S_i,A_i,W_i,V_i - X_{1,i}$ -
\begin{eqnarray}
    p(x_{1,i}|m_{12},x_1^{i-1},a_i,s_{i}^n,y^{i-1},m_2,x_{2,i}) &=& \sum_{m_1\in \mathcal{M}_1} p(m_1|m_{12},x_1^{i-1},a_i,s_{i}^n,y^{i-1},m_2,x_{2,i})\notag\\&&\tab p(x_{1,i}|m_1,m_{12},x_1^{i-1},a_i,s_{i}^n,y^{i-1},m_2,x_{2,i})\\
    &\stackrel{(a)}=&\sum_{m_1\in \mathcal{M}_1} p(m_1|m_{12},x_1^{i-1},a_i,s_i)\notag\\&&\tab p(x_{1,i}|m_1,m_{12},x_1^{i-1},a_i,s_i)\\
    &=&p(x_{1,i}|m_{12},x_1^{i-1},a_i,s_i),
\end{eqnarray}
where (a) follows since $m_1$ is independent of $(s_{i+1}^n,y^{i-1},m_2,x_{2,i})$ given $(m_{12},x_1^{i-1},a_i,s_i)$ and since $x_{1,i} = f(m_1)$.
\item $Y_i - (X_{1,i},X_{2,i},S_i) - (W_i,V_i,U_i,A_i)$ - Follows from the fact that the channel output at any time $i$ is assumed to depend only on the channel inputs and state at time $i$.
\end{itemize}
Finally, let $Q$ be an RV independent of $(X_1^n,X_2^n,Y^n)$ and uniformly distributed over the set $\{1,2,3,\dots,n\}$. We define the RV $W\triangleq(Q,W_Q)$ and obtain the region given in (\ref{eq:capacity8}).\hfill $\blacksquare$

\textit{Converse for the causal case:}
For the causal case we repeat the same approach as for the strictly causal case, except that in the final step we need to show the Markov chain $U_i - (S_i,A_i,W_i,V_i) - X_{1,i}$. We can see from the following derivations that this Markov chain holds
\begin{eqnarray}
    p(x_{1,i}|m_{12},x_1^{i-1},a_i,s_{i}^n,y^{i-1},m_2) &=& \sum_{m_1\in \mathcal{M}_1} p(m_1|m_{12},x_1^{i-1},a_i,s_{i}^n,y^{i-1},m_2)\notag\\&&\tab p(x_{1,i}|m_1,m_{12},x_1^{i-1},a_i,s_{i}^n,y^{i-1},m_2)\\
    &\stackrel{(a)}=&\sum_{m_1\in \mathcal{M}_1} p(m_1|m_{12},x_1^{i-1},a_i,s_i)\notag\\&&\tab p(x_{1,i}|m_1,m_{12},x_1^{i-1},a_i,s_i)\\
    &=&p(x_{1,i}|m_{12},x_1^{i-1},a_i,s_i),
\end{eqnarray}
where (a) follows since $m_1$ is independent of $(s_{i+1}^n,y^{i-1},m_2)$ given $(m_{12},x_1^{i-1},a_i,s_i)$ and since $x_{1,i} = f(m_1)$.
\hfill $\blacksquare$

\subsection{Achievability}
\textit{Achievability for the strictly causal case:}
Fix a joint distribution $P(w)P(v|w)P(a|w)P(s|a)P(x_1|v,w)P(u|s,w,v,a)$\newline$p(x_2|w,a,v,u,s)P(y|x_1,x_2,s)$ where $P(s|a)$ and $P(y|x_1,x_2,s)$ are given by the channel. In the following achievability scheme we use block Markov coding, rate splitting, and Gelfand-Pinsker coding.

\emph{\underline{Coding Scheme:}} We consider $B$ blocks, each consisting of $n$ symbols; thus we transmit $nB$ symbols. We transmit $B-1$ messages $M_1$ in B blocks of information. Here, $M_1 \in \{1,\dots,2^{nR_1}\}$; thus, asymptotically, for a large enough $n$, our transmission rate would be $\frac{nR_1(B-1)}{nB}\stackrel{n\rightarrow\infty}\longrightarrow R_1$. We also split message $M_1$ into $(M_1^\prime,M_1^{\prime\prime})$ such that $(R_1^\prime,R_1^{\prime\prime}) = (C_{12},R_1 - C_{12})$.

\emph{\underline{Code Design:}} Generate $2^{nR_1^\prime}$ codewords $w^n$ i.i.d. using $P(w^n)=\Pi_{i=1}^nP(w_i)$. For each $w^n$, generate $2^{nR_1^{\prime\prime}}$ codewords $v^n$ i.i.d. using $P(v^n|w^n)=\Pi_{i=1}^nP(v_i|w_i)$. For each $w^n$, generate $2^{nR_2}$ codewords $a^n$ i.i.d. using $P(a^n|w^n)=\Pi_{i=1}^nP(a_i|w_i)$. For each pair $(w^n,v^n)$, generate $2^{nR_1^{\prime\prime}}$ codewords $x_1^n$ i.i.d. using $P(x_{1}^n|v^n,w^n)=\Pi_{i=1}^nP(x_{1,i}|v_i,w_i)$. Additionally, for each triplet $(w^n,v^n,a^n)$, generate $2^{n(R_2+\tilde{R})}$ codewords $u^n$ i.i.d. using $P(u^n|a^n,v^n,w^n)=\Pi_{i=1}^nP(u_i|a_i,v_i,w_i)$. Randomly bin all $u^n$ codewords into $2^{nR_2}$ bins where each bin contains $2^{n\tilde{R}}$ codewords.

\emph{\underline{Encoding:}} We denote the realizations of the messages $(M_1^\prime,M_1^{\prime\prime},M_2)$ at block $b$ as $(m_{1,b}^{\prime},m_{1,b}^{\prime\prime},m_{2,b})$.
Since we use block Markov coding, we set $m_{1,B}=1$.
In block $b \in \{1,\dots,B\}$, send $m_{1,b}^{\prime}$ from Encoder 1 to Encoder 2 via the rate-limited cooperation link. Encode message $m_{1,b}^\prime$ using $w^n(m_{1,b}^\prime)$. Encode message $m_{1,b-1}^{\prime\prime}$ conditioned on $m_{1,b}^\prime$ using $v^n(m_{1,b-1}^{\prime\prime},w^n)$ and encode message $m_{1,b}^{\prime\prime}$ conditioned on $(m_{1,b-1}^{\prime\prime},m_{1,b}^\prime)$ using $x_1^n(m_{1,b}^{\prime\prime},v^n,w^n)$. Given $(m_{1,b}^\prime,m_{2,b})$, Encoder 2 chooses an action sequence $a^n$. Given $(s^n,w^n,v^n,a^n)$, look in bin $m_{2,b}$ for a codeword $u^n(w^n,v^n,a^n,m_{2,b},l)$ that is jointly typical with $(w^n(m_{1,b}^\prime),v^n(m_{1,b-1}^{\prime\prime}),s^n,a^n(m_{2,b}))$, where $l\in\{1,\dots,2^{n\tilde{R}}\}$. Send $x_1^n(m_{1,b}^{\prime\prime},w^n,v^n)$ and $x_2^n$ according to $p(x_2|w,v,u,s)$ i.i.d. over the channel.

\emph{\underline{Decoding at Encoder 2:}} At the end of block $b$, Encoder 2 tries to decode message $m^{\prime\prime}_{1,b}$. Given $m^\prime_{1,b}$ and assuming that message $m^{\prime\prime}_{1,b-1}$ was decoded correctly at the end of block $b-1$, Encoder 2 looks for $\hat{m}^{\prime\prime}_{1,b}$ s.t.
\begin{equation}
    (w^n(m^\prime_{1,b}),v^n(m^{\prime\prime}_{1,b-1},w^n),x_1^n(\hat{m}^{\prime\prime}_{1,b},w^n,v^n))\in T_\epsilon^{(n)}(W,V,X_1).
\end{equation}
If no such $\hat{m}^\prime_{1,b}$, or more than one such $\hat{m}^\prime_{1,b}$, was found, an error is declared at block $b$ and therefore in the whole super-block $nB$.

\emph{\underline{Decoding at the receiver:}} At the end of block $B$, the decoding is done backwards. At block $b$, assuming that $m_{1,b}$ was decoded correctly in block $b+1$, the decoder looks for the triplet $(m_{1,b}^{\prime},m_{1,b-1}^{\prime\prime},\hat{m}_{2,b})$ s.t.
\begin{equation}
    (w^n(\hat{m}_{1,b}^{\prime}),v^n(\hat{m}_{1,b-1}^{\prime\prime},w^n),x_1^n(m_{1,b}^{\prime\prime},w^n,v^n),a^n(\hat{m}_{2,b},w^n),u^n(\hat{m}_{2,b},w^n,v^n,s^n,a^n,l),y^n)\in T_\epsilon^{(n)}(W,V,X_1,A,U,Y).
\end{equation}
If no such pair, or more than one such pair, was found, an error is declared at block $b$ and therefore in the whole super-block $nB$.

\emph{\underline{Error Analysis:}} Without loss of generality, we assume that $(m_{1,b}^{\prime},m_{1,b-1}^{\prime\prime},m_{2,b}) = (1,1,1)$.
The probability that $x_1^n(1,w^n,v^n)=x_1^n(i,w^n,v^n)$ where $i>1$ and where $(w^n(1),v^n(1,w^n),x_1^n(1,w^n,v^n))\in T_\epsilon^{(n)}(W,V,X_1)$ is bounded by $2^{-n(H(X_1|V,W)-\delta(\epsilon))}$, where $\delta(\epsilon)$ goes to zero as $\epsilon$ goes to zero. Hence, if
\begin{equation}\label{eq:action-E1}
    R_1-C_{12}<H(X_1|V,W),
\end{equation}
then the probability that an incorrect message $m_{1,b}$ was decoded goes to zero for a large enough $n$.
We define the following event at block $b$:
\begin{equation}
    E_{i,j,k,l,b} \triangleq (w^n(i),v^n(j,w^n),x_1^n(\hat{m}_{1,b}^{\prime\prime},v^n,w^n),a^n(k,w^n),u^n(k,v^n,s^n,a^n,w^n,l),y^n)\in T_\epsilon^{(n)}(W,V,X_1,A,U,Y).
\end{equation}
We can bound the probability of error as follows:
\begin{eqnarray}
    P_{e,b}^{(n)} &\leq& \Pr(E^c_{1,1,1,1,b}) + \sum_{\substack{{i=1,j=1}\\{k>1,l>1}}}\Pr(E_{1,1,k,l,b}) +  \sum_{\substack{{i=1,j>1}\\{k=1,l>1}}}\Pr(E_{1,j,1,l,b})\notag\\
    && + \sum_{\substack{{i=1,j>1}\\{k>1,l>1}}}\Pr(E_{1,j,k,l,b}) + \sum_{\substack{{i>1,j>1}\\{k>1,l>1}}}\Pr(E_{i,j,k,l,b}). \label{eq:prob_of_error2}
\end{eqnarray}
We now show that each term in (\ref{eq:prob_of_error2}) goes to zero for a large enough $n$.
\begin{itemize}
  \item Upper-bounding $\Pr(E^c_{1,1,1,1,b})$: Since we assume that Transmitters 1 and 2 encode the correct message triplet $(m_{1,b}^{\prime},m_{1,b-1}^{\prime\prime},m_{2,b})$ at block $b$ and that the receiver decoded the right $(m_{1,b+1}^{\prime},m_{1,b}^{\prime\prime},m_{2,b+1})$ at block $b+1$, by the LLN, $\Pr(E^c_{1,1,1,b}) \rightarrow 0$.
  \item Upper-bounding $\sum_{\substack{{i=1,j=1}\\{k>1,l>1}}}\Pr(E_{1,1,k,l,b})$: Assuming that $m_{1,b}^{\prime\prime}$ was decoded correctly at block $b+1$, the probability for this event is bounded by
      \begin{equation}
            \sum_{\substack{{i=1,j=1}\\{k>1,l>1}}}\Pr(E_{1,1,k,l,b}) \leq 2^{n(R_2+\tilde{R})}2^{-n(I(U,A;Y|W,V,X_1)-\delta(\epsilon)}. \label{eq:action-E2}
      \end{equation}
  \item Upper-bounding $\sum_{\substack{{i=1,j>1}\\{k=1,l>1}}}\Pr(E_{1,j,1,l,b})$: Similarly to (\ref{eq:action-E2}) we obtain
      \begin{equation}
            \sum_{\substack{{i=1,j>1}\\{k=1,l>1}}}\Pr(E_{1,j,1,l,b}) \leq 2^{n(R_1-C_{12}+\tilde{R})}2^{-n(I(V,X_1,U;Y|W,A)-\delta(\epsilon)}. \label{eq:action-E3}
      \end{equation}
  \item Upper-bounding $\sum_{\substack{{i=1,j>1}\\{k>1,l>1}}}\Pr(E_{1,j,k,l,b})$: Similarly to (\ref{eq:action-E2}) we obtain
      \begin{equation}
            \sum_{\substack{{i=1,j>1}\\{k>1,l>1}}}\Pr(E_{1,j,k,l,b}) \leq 2^{n(R_1-C_{12}+R_2+\tilde{R})}2^{-n(I(U,A,V,X_1;Y|W)-\delta(\epsilon)}. \label{eq:action-E4}
      \end{equation}
  \item Upper-bounding $\sum_{\substack{{i>1,j>1}\\{k>1,l>1}}}\Pr(E_{1,j,k,l,b})$: Similarly to (\ref{eq:action-E2}) we obtain
      \begin{equation}
            \sum_{\substack{{i>1,j>1}\\{k>1,l>1}}}\Pr(E_{1,j,k,l,b}) \leq 2^{n(R_1+R_2+\tilde{R})}2^{-n(I(U,A,W,V,X_1;Y)-\delta(\epsilon)}. \label{eq:action-E5}
      \end{equation}
\end{itemize}
Finally, we analyze the probability of error for finding $u^n$ at Encoder 2. By the covering lemma, if
\begin{equation}
    \tilde{R}>I(U;S|W,V,A)  \label{eq:action-E6}
\end{equation}
then with high probability, in block $b$ we can find a codeword $u^n$ that is jointly typical with $s^n$ in bin number $m_{2,b}$.
The combination of (\ref{eq:action-E1}), (\ref{eq:action-E2}), (\ref{eq:action-E3}), (\ref{eq:action-E4}), (\ref{eq:action-E5}), and (\ref{eq:action-E6}) yields the capacity region in (\ref{eq:capacity8}), thus completing the proof.\hfill $\blacksquare$

\textit{Achievability for the causal case:}
The achievability part follows similarly to that of the strictly causal case, only now the generation of $X_2^n$ is done i.i.d. according to the conditional distribution of $p(x_2|w,v,u,s,x_1)$ induced by (\ref{eq:action-joint-B}).
\hfill $\blacksquare$

\bibliographystyle{IEEEtran}
\bibliography{MAC+cribbing+conferencing}

\nocite{somekh2008cooperative}
\nocite{cover2012elements}
\nocite{goldfeld2013finite}
\end{document}